\documentclass{revtex4}
\usepackage{amsmath,amsthm}
\usepackage{rotating}
\usepackage{multirow}
\usepackage{graphicx}

\usepackage{amsfonts}
\newtheorem{theorem}{Theorem}[section]
\newtheorem{lemma}[theorem]{Lemma}
\newtheorem{proposition}[theorem]{Proposition}
\newtheorem{cor}[theorem]{Corollary}

\theoremstyle{remark}
\newtheorem{remark}[theorem]{Remark}

\theoremstyle{definition}
\newtheorem{definition}[theorem]{Definition}

\theoremstyle{example}
\newtheorem{example}[theorem]{Example}

\theoremstyle{notation}

\newcommand{\bra}[1]{\langle#1|}
\newcommand{\ket}[1]{|#1\rangle}

\begin{document}

\title{ Dressed coherent states in finite quantum systems: a cooperative game theory approach }            
\author{A. Vourdas}
\affiliation{Department of Computer Science,\\
University of Bradford, \\
Bradford BD7 1DP, United Kingdom\\a.vourdas@bradford.ac.uk}

\begin{abstract}
A quantum system with variables in ${\mathbb Z}(d)$ is considered.
Coherent density matrices and coherent projectors of rank $n$ are introduced, and their properties (e.g., the resolution of the identity)
are discussed. Cooperative game theory and in particular the Shapley methodology, is used to renormalize coherent states, 
into a particular type of coherent density matrices (dressed coherent states).
The $Q$-function of a Hermitian operator, is then renormalized into a physical analogue of the Shapley values. 
Both the $Q$-function and the Shapley values, are used to study 
the relocation of a Hamiltonian in phase space as the coupling constant varies, and its effect on the ground state of the system.
The formalism is also generalized for any total set of states, for which we have no resolution of the identity.
The dressing formalism leads to density matrices that resolve the identity, and makes them practically useful. 

\end{abstract}
\maketitle

\section{Introduction}

Coherent states  have been studied extensively in the literature\cite{coh1,coh2,coh3}.
They are an overcomplete set of states, and
there are subsets of the full set of coherent states which are total sets (i.e., there is no state which is orthogonal to all states in the subset).  
We consider quantum systems with $d$-dimensional Hilbert space\cite{FIN,FIN1,FIN2}, in which case the number of coherent states is $d^2$\cite{COH,COH1}.
We show that ideas from cooperative game theory can provide a deeper insight to the
overcompleteness of coherent states, and their linear dependence (lack of linear independence) which is related to it.

Cooperative game theory \cite{S0,S1,S2,S3} adds `corrections\rq{} to the individual contribution of a player,  
which reflect his contribution to coalitions (aggregations) of players.
This gives the Shapley value, which shows the share of a player in the `total worth\rq{} of the game, and which renormalizes (dresses) his lone contribution, 
by adding his contribution to all possible coalitions.
The sum of all the Shapley values, is the total worth of the game.

In analogy to this,  we add to the one-dimensional projector corresponding to a coherent state, other terms that reflect its role within 
spaces spanned by aggregations of coherent states.
This gives a renormalized (dressed) coherent state, which is a mixed state described by a coherent density matrix.
The sum of all coherent density matrices is the identity (resolution of the identity).

Coherent states are used to define the $Q$-function of a Hermitian operator $\theta$.
The dressed coherent states define a generalized $Q$-function, a physical analogue of the Shapley values.
As an application of these ideas, we study the relocation of a Hamiltonian in phase space as a function of
the coupling constant, and how this affects the ground state of the system. 
Such calculations on large quantum systems, can be used in the study of phase transitions.
Location indices of Hermitian operators in phase space, and comonotonicity (or cohabitation) intervals of the coupling constant, 
are used to quantify this relocation.
These calculations are performed with respect to either $Q$-function or Shapley values,
and their relative merits are discussed.

Most of the paper uses the formalism with coherent states. But we also consider a total set of states, 
which are not coherent states and for which we have no resolution of the identity.
The renormalization formalism in this case, leads to density matrices that resolve the identity and can be used in practical applications. 

In order to develop this formalism, we introduce in section \ref{cohden} coherent density matrices.
They are generalizations of coherent states (which are pure states) to mixed states.
We then introduce in section \ref{eac} coherent projectors of rank $n$, to spaces spanned by aggregations of $n$ coherent states.
The terms coherent density matrices and coherent projectors, reflect the fact that they resolve the identity, and that there is a closure property 
where under displacements they are transformed into other coherent density matrices and coherent projectors, correspondingly. 
 
In section \ref{SH} we present briefly some concepts from cooperative game theory, which are needed later.
The presentation uses the standard language of cooperative game theory, but it also 
introduces some `quantum terminology\rq{}, because our aim is to transfer these ideas in a quantum context. 
In section \ref{IV} we explain in a precise manner, the analogies between cooperative game theory, and aggregations of coherent states.
M\"obius transformations are used to identify overlaps and avoid double-counting:
in cooperative game theory, they quantify the added value in a coalition; and
in a quantum context they describe the double counting due to the overcompleteness of the coherent states. 

In section \ref{V} we transfer the concept of Shapley values into a quantum context.
We show that they are generalized $Q$-functions with respect to a particular set of coherent density matrices, 
which can be regarded as renormalized (dressed) versions of the `bare' coherent projectors.
The dressing formalism is related to the non-orthogonality and non-commutativity of the coherent projectors, as discussed in section \ref{VI}.
The properties of these coherent density matrices are presented in propositions \ref{V2} and \ref{cxz}.

In section \ref{gen} the formalism is generalized to any total set of states.
As an application we study in section \ref{appl}, 
the relocation of a Hamiltonian in phase space as the coupling constant varies, and
its effect on the ground state of the system.
We conclude is section \ref{D} with a discussion of our results.

\section{Generalized coherence in finite quantum systems}
\subsection{Coherent projectors}\label{vv}

We consider a quantum system with variables in ${\mathbb Z}(d)$, described by a $d$-dimensional Hilbert space $H(d)$.
$\ket{X;n}$ is the basis of position states, and $\ket{P;n}$ the basis of momentum states
($X$ and $P$ in the notation are not variables, they simply indicate position and momentum states).
They are related through a finite Fourier transform:\cite{FIN}:
\begin{eqnarray}
&&F=d^{-1/2}\sum _m\omega (mn)\ket{X;n}\bra{X:m};\;\;\;\;\omega(\alpha )=\exp \left(\frac{i2\pi \alpha}{d}\right )
\nonumber\\
&&\ket{P;n}=F\ket{X;n};\;\;\;\;m,n, \alpha\in {\mathbb Z}(d).
\end{eqnarray}
Displacement operators in the ${\mathbb Z}(d)\times {\mathbb Z}(d)$ phase space, are given by
\begin{eqnarray}\label{dis}
D(\alpha , \beta)=Z^{\alpha}X^{\beta}\omega (-2^{-1}\alpha \beta);\;\;\;\;
Z=\sum _m\omega (m)\ket{X;m}\bra{X;m};\;\;\;\;
X=\sum _m \ket{X;m+1}\bra{X;m}
\end{eqnarray}
We consider the case where $d$ is an odd integer (in this case the $2^{-1}$ exists in ${\mathbb Z}(d)$). 

Acting with $D(\alpha , \beta)$ on a (normalized) fiducial vector $\ket {\eta}$,
we get the $d^2$ coherent states\cite{COH,COH1}:
\begin{eqnarray}\label{coh}
&&\ket{C;\alpha, \beta}=D(\alpha , \beta)\ket{\eta};\;\;\;\;\alpha, \beta \in {\mathbb Z}(d)\nonumber\\
&&\ket {\eta}=\sum _m \eta _m\ket{X;m};\;\;\;\;\sum _m|\eta_m|^2=1.
\end{eqnarray}
The $C$ in the notation indicates coherent states.
The fiducial vector is called `generic', if any subset of $d$ coherent states, from the corresponding $d^2$ coherent states, are linearly independent.
Position and momentum states, are examples of non-generic fiducial vectors.  
Below we use generic fiducial vectors.
In some cases, we show explicitly the fiducial vector in the notation of coherent states, as $\ket{C(\eta);\alpha, \beta}$.

Let $H(\alpha,\beta)$ be the one-dimensional subspace  that contains the coherent states $\ket{C;\alpha,\beta}$,
and $\Pi(\alpha,\beta)$ be the corresponding `coherent projector'. Then
\begin{eqnarray}\label{1111}
&&\frac{1}{d}\sum _{\alpha,\beta}\Pi({\alpha,\beta})={\bf 1};\;\;\;\;\;\;\Pi ({\alpha,\beta})=\ket{C;\alpha,\beta}\bra{C;\alpha,\beta}\nonumber\\
&&D(\gamma, \delta)\Pi({\alpha ,\beta })D^{\dagger}(\gamma, \delta)=\Pi({\alpha +\gamma,\beta +\delta})
\end{eqnarray}
The first relation is a resolution of the identity. The second relation is a closure property, where
under displacement transformations, these projectors are transformed into other projectors of the same type.
The term `coherent\rq{} refers to these two properties.
Using the resolution of the identity, we expand an arbitrary state $\ket{s}$ in terms of coherent states as
\begin{eqnarray}\label{rt}
\ket{s}=\sum _{\alpha,\beta}s(\alpha, \beta)\ket{C;\alpha,\beta};\;\;\;\;\;s(\alpha, \beta)=\frac{1}{d}\bra{C;\alpha,\beta}s\rangle.
\end{eqnarray}

Given a Hermitian operator $\theta$, its $Q$-function is a real function given by
\begin{eqnarray}\label{mmm}
Q(\alpha,\beta\;|\;\theta)=\frac{1}{d}{\rm Tr}[\Pi (\alpha,\beta) \theta];\;\;\;\;\;\;\sum _{\alpha, \beta} Q(\alpha,\beta\;|\;\theta)={\rm Tr}(\theta).
\end{eqnarray}
In the case of a positive semi-definite operator $\theta$, its $Q$-function is non-negative.

Sometimes we will replace the `pair of indices notation\rq{}, with a `single index notation\rq{}, and denote the 
coherent states as $\ket{C; i}$ and the
coherent projectors as $\Pi(i)$, where $i$ takes the values in the set
\begin{eqnarray}\label{51}
\Omega=\{1,...,d^2\}.
\end{eqnarray}
We use here a bijective lexicographic map between ${\mathbb Z}(d)\times {\mathbb Z}(d)$ and $\Omega$, as follows:
\begin{eqnarray}\label{nota}
&&(0,0)\;\rightarrow\;1;\;\;\;(0,1)\;\rightarrow\;2;...;(0,d-1)\;\rightarrow\;d\nonumber\\
&&(1,0)\;\rightarrow\;d+1;...;(1,d-1)\;\rightarrow\;2d\nonumber\\
&&..............\nonumber\\
&&(d-1,0)\;\rightarrow\;d^2-d+1;...;(d-1,d-1)\;\rightarrow\;d^2
\end{eqnarray}

If $\theta _{mn}=\bra{X;m}\theta\ket{X;n}$
then
\begin{eqnarray}\label{xxx}
Q(\alpha,\beta|\theta)=\sum M(\alpha, \beta;m,n)\theta _{mn};\;\;\;\;M(\alpha, \beta;m,n)=\frac{1}{d}\bra{C;\alpha, \beta}X;m\rangle \langle X;n\ket{C;\alpha, \beta}
\end{eqnarray}
In the single index notation, we rewrite $ \theta _{mn}$ as $\theta _j$, and the $M(\alpha, \beta;m,n)$ as ${\mathfrak M}_{ij}$,  and we get
\begin{eqnarray}\label{xxxz}
Q(i|\theta)=\sum {\mathfrak M}_{ij} \theta _{j};\;\;\;i,j\in {\Omega}.
\end{eqnarray}
Both the Hermitian operator $\theta$ and the $Q$-function, are represented as vectors in a $d^2$-dimensional space.
Given the $Q(i|\theta)$,
Eq.(\ref{xxxz}) is a system of $d^2$ equations with $d^2$ unknowns, which can be used to calculate the $d^2$ values of $\theta _j$, i.e., the operator $\theta$.
The generic nature of the fiducial vector ensures that the $d^2\times d^2$ matrix ${\mathfrak M}_{ij}$ is invertible.

\subsection{Coherent density matrices}\label{cohden}
In the previous section, the fiducial state of coherent states, was a pure state.
We now consider the case where the fiducial state is a mixed (in general) state.
\begin{definition}\label{def}
Let $R_0$ be a density matrix, which we call fiducial density matrix.
The $d^2$ density matrices
\begin{eqnarray}\label{6v}
R (\alpha, \beta)=D(\alpha, \beta)R_0 D^{\dagger}(\alpha, \beta),
\end{eqnarray}
are coherent density matrices.

\end{definition}
If the fiducial density matrix $R_0$ is a projector describing a pure state, then the coherent density matrices reduce to the coherent projectors of the previous section.
The term coherent density matrices expresses the following properties which are analogous to Eq.(\ref{1111}):
\begin{proposition}
\begin{eqnarray}\label{11q}
&&\frac{1}{d}\sum _{\alpha,\beta}R({\alpha,\beta})={\bf 1}\label{5x}\\
&&D(\gamma, \delta)R({\alpha ,\beta })D^{\dagger}(\gamma, \delta)=R({\alpha +\gamma,\beta +\delta})\label{5c}
\end{eqnarray} 
\end{proposition}
\begin{proof}
For any operator $\Phi _0$ we define the 
\begin{eqnarray}\label{6v}
\Phi (\alpha, \beta)=D(\alpha, \beta)\Phi _0 D^{\dagger}(\alpha, \beta).
\end{eqnarray}
It is known that (Eq.(119) in \cite{Fin2})
\begin{eqnarray}
\frac{1}{d}\sum _{\alpha,\beta}\Phi({\alpha,\beta})={\bf 1}{\rm Tr}(\Phi_0)
\end{eqnarray} 
Eq.(\ref{5x}) is a special case of this.
Analogous relation also holds in infinite dimensional Hilbert spaces (Eq.(38) in \cite{inf}). 

In order to prove Eq.(\ref{5c}), we multiply Eq.(\ref{6v}) with $D(\gamma, \delta)$ on the left and $D^{\dagger}(\gamma, \delta)$ on the right,
and use the relation
\begin{eqnarray}
D(\gamma, \delta)D(\alpha, \beta)=D(\gamma +\alpha, \delta +\beta)\omega[2^{-1}(\gamma \beta -\delta \alpha)]
\end{eqnarray} 
\end{proof}

We express $R_0$ in terms of its eigenvalues and eigenvectors, as
\begin{eqnarray}\label{6v}
R_0=\sum _{n=1}^d\lambda _n\ket{e_n}\bra{e_n};\;\;\;\;\sum _{n=1}^d\lambda _n=1
\end{eqnarray}
$R_0$ is a collection of the orthogonal pure states $\ket{e_n}$, with the probabilities $\lambda _n$ attached to them.
Also 
\begin{eqnarray}\label{6v}
R (\alpha, \beta)=\sum _{n=1}^d\lambda _n\ket{C(e_n); \alpha, \beta}\bra{C(e_n); \alpha, \beta}
\end{eqnarray}
Here the $e_n$ in the notation $\ket{C(e_n); \alpha, \beta}$, indicates explicitly the fiducial vector. 
It is seen that in the coherent states we use a single fiducial vector, while in the coherent density matrices we use
$d$ orthogonal fiducial vectors $\ket{e_n}$ with probabilities $\lambda _n$, described by the fiducial density matrix $R_0$.
The set of coherent density matrices $\{R (\alpha, \beta)\}$ describes the $d$ sets of coherent states 
$\Sigma _n=\{\ket{C(e_n); \alpha, \beta}|\alpha, \beta \in {\mathbb Z}(d)\}$, with probabilities $\lambda _n$.
The fiducial density matrix $R_0$ is called generic, if its eigenvalues are different from each other, and if for a fixed $n$ any $d$ of the $d^2$ coherent states $\{\ket{C(e_n); \alpha, \beta}\}$, are linearly independent.

An arbitrary state can be expanded as
\begin{eqnarray}
\ket {s}=\frac{1}{d}\sum _{\alpha,\beta}R({\alpha,\beta})\ket{s}=\sum _{n=1}^d\sum _{\alpha,\beta}s_n(\alpha, \beta)\ket{C(e_n); \alpha, \beta};\;\;\;\;
s_n(\alpha, \beta)=\frac{\lambda _n}{d}\bra{C(e_n); \alpha, \beta}s\rangle
\end{eqnarray} 
Comparison with Eq.(\ref{rt}), shows that 
the coherent states are a single overcomplete basis, while coherent density matrices are a set of $d$ overcomplete bases, with probabilities attached to them.

Given a Hermitian operator $\theta$, its $Q$-function with respect to the coherent density matrices $R (\alpha,\beta)$, is given by
\begin{eqnarray}\label{mmm1}
Q_R(\alpha,\beta|\theta)=\frac{1}{d}{\rm Tr}[R (\alpha,\beta) \theta].
\end{eqnarray}
It is easily seen that
\begin{eqnarray}\label{mmm11}
\sum _{\alpha, \beta} Q_R(\alpha,\beta|\theta)={\rm Tr}(\theta);\;\;\;\;
Q_R(\alpha,\beta|\lambda _1\theta _1+\lambda _2\theta _2)=\lambda _1Q_R(\alpha,\beta|\theta _1)
+\lambda _2Q_R(\alpha,\beta|\theta _2)
\end{eqnarray}
Using Eq.(\ref{6v}) we rewrite this as
\begin{eqnarray}
Q_R(\alpha,\beta|\theta)=\sum _{n=1}^d\lambda _nQ_n(\alpha,\beta\;|\;\theta);\;\;\;
Q_n(\alpha,\beta|\theta)=\frac{1}{d}\bra{C(e_n); \alpha, \beta}\theta \ket{C(e_n); \alpha, \beta}.
\end{eqnarray}
We note that the $d$ states in the set $\Sigma (\alpha, \beta)=\{\ket{C(e_n); \alpha, \beta}|n=1,...,d\}$ are orthogonal to each other, and 
\begin{eqnarray}
d\sum _{n=1}^dQ_n(\alpha,\beta|\theta)={\rm Tr}(\theta).
\end{eqnarray}

\subsection{Coherent projectors of rank $n$ describing aggregations of $n$ coherent states}\label{eac}

Let $A=\{(\alpha _1,\beta _1);...;(\alpha _n,\beta _n)\}$  be a subset of $\Omega$.
We denote as $|A|$ the cardinality of the set $A$.
We consider the subspace of $H(d)$:
\begin{eqnarray}
H(A)=H({\alpha _1,\beta _1};...;{\alpha _n,\beta _n})={\rm span}[H({\alpha _1,\beta _1}) \cup ...\cup H({\alpha _n,\beta _n})]
\end{eqnarray}
It contains all superpositions
$\kappa _1\ket {C;{\alpha _1,\beta _1}}+...+\kappa _n \ket{C;{\alpha _n,\beta _n}}$, and in this sense it describes the aggregation of these coherent states.
We denote the corresponding projector to this subspace, as $\Pi({\alpha _1,\beta _1};...;{\alpha _n,\beta _n})$ or as $\Pi(A)$.
In the single index notation, these projectors are $\Pi(i_1,...,i_n)$ and the corresponding subspace $H(i_1,...,i_n)$.
We note that the  order of the indices is not important. 
Since we consider generic fiducial vectors, if $|A|\le d$, the coherent states are linearly independent, and the $H(A)$ is $|A|$-dimensional.
Also
\begin{eqnarray}\label{g7}
|A|\ge d\;\;\rightarrow\;\;\Pi(A)={\bf1}.
\end{eqnarray}
We note that
\begin{eqnarray}\label{123}
\Pi(A)\ne \sum _{i\in A}\Pi (i).
\end{eqnarray}
The aggregation is different than the sum of its parts. Later we will see an analogous relation for coalitions in cooperative game theory (Eq.(\ref{149})).

We use the shorthand notation
\begin{eqnarray}
A+(\gamma, \delta)=\{(\alpha _1+\gamma,\beta _1+\delta);...;(\alpha _n+\gamma,\beta _n+\delta)\}.
\end{eqnarray}
The $\Pi(A)$ are coherent projectors in the sense of the properties in the following proposition: 
\begin{proposition}
\begin{itemize}
\item[(1)]
The projectors  $\Pi(A)$ resolve the identity as follows:
\begin{eqnarray}\label{45}
\frac{1}{d|A|}\sum _{\gamma, \delta}\Pi[A+(\gamma, \delta)]={\bf 1};\;\;\;|A|<d,
\end{eqnarray}
\item[(2)]
The following `closure property\rq{}, holds:
\begin{eqnarray}\label{44}
D(\gamma, \delta)\Pi(A)D^{\dagger}(\gamma, \delta)=
\Pi[A+(\gamma, \delta)].
\end{eqnarray}
\end{itemize}
\end{proposition}
\begin{proof}
We gave the proof in \cite{v16a}. 
\end{proof}
These two properties are generalizations of those in Eq.(\ref{1111}), and they justify the name `coherent projectors of rank $n$'.
Using the resolution of the identity in Eq.(\ref{45}), we can express an arbitrary state $\ket{s}$ as
\begin{eqnarray}
\ket{s}=\frac{1}{d|A|}\sum _{\gamma, \delta}\Pi[A+(\gamma, \delta)]\ket{s};\;\;\;|A|<d.
\end{eqnarray}
This expresses $\ket{s}$ as a sum of $d^2$ vectors in the spaces $H[A+(\gamma, \delta)]$, and is a generalization of Eq.(\ref{rt}). 
\begin{cor}
The following resolution of the identity involves all $\begin{pmatrix}d^2\\k\\\end{pmatrix}$ subsets of $\Omega$, with given cardinality $|A|=k$:
\begin{eqnarray}\label{c45}
\frac{1}{d}\begin{pmatrix}d^2-1\\k-1\\\end{pmatrix}^{-1}\sum _{|A|=k}\Pi(A)={\bf 1};\;\;\;k<d,
\end{eqnarray}
\end{cor}
\begin{proof}
There are $d^2$ terms in the resolution of the identity in Eq.(\ref{45}). Therefore there are 
\begin{eqnarray}
r=\frac{1}{d^2}\begin{pmatrix}d^2\\k\\\end{pmatrix}
\end{eqnarray}
resolutions of the identity, that involve sets of indices with given cardinality $|A|=k$.
We add all of them and we get the resolution of the identity in Eq.(\ref{c45}).  
$rk$ is equal to $\begin{pmatrix}d^2-1\\k-1\\\end{pmatrix}$.
\end{proof}
\begin{lemma}\label{fg}
Let
\begin{eqnarray}\label{1cd}
g_{ij}(A)=\langle C;i\ket{C;j};\;\;\;i,j\in A\subseteq \Omega.
\end{eqnarray}
$g_{ij}(A)$ is an $|A|\times |A|$ Hermitian matrix of rank $\min(|A|,d)$. 
In the case $|A|\le d$ its inverse
$G(A)=[g(A)]^{-1}$ exists, and it is  an $|A|\times |A|$ Hermitian positive definite matrix.
\end{lemma}
\begin{proof}
If $|A|\le d$, an arbitrary (normalized) state $\ket{s}$ in the $|A|$-dimensional space $H(A)$, can be written as 
\begin{eqnarray}
\ket{s}=\sum _{i}s_i\ket{C;i};\;\;\;i\in A
\end{eqnarray}
The coherent states $\ket{C;i}$ with $i\in A$, are linearly independent because we use a generic fiducial vector.
Then
\begin{eqnarray}
\langle s\ket{s}=\sum _{i,j}s_i^*s_jg_{ij}(A)=1.
\end{eqnarray}
This shows that the $|A|\times |A|$ Hermitian matrix $g_{ij}(A)$ is positive definite. Therefore it is invertible, and 
its inverse $G_{ij}(A)$ is also a Hermitian positive definite matrix.

In the case $|A|>d$, $H(A)$ is the full space $H(d)$.
The corresponding coherent states $\ket{C;i}$ with $i\in A$, are not linearly independent and the rank of the matrix $g_{ij}(A)$ is $d$.
Therefore in this case the $g_{ij}(A)$ is not invertible.
\end{proof}

\begin{proposition}
Let $A$ be a subset of $\Omega$ with $|A|\le d$. The projector $\Pi(A)$ can be expressed as follows:
\begin{eqnarray}\label{cf1}
\Pi(A)=\sum _{i,j}G_{ij}(A)\ket{C;i}\bra{C;j};\;\;\;i,j\in A.
\end{eqnarray}
The $G_{ij}(A)$ has been defined in lemma \ref{fg}.
In the case that the cardinality of $A$ is $d$, 
\begin{eqnarray}\label{rr}
\sum _{i,j}G_{ij}(A)\ket{C;i}\bra{C;j}={\bf 1};\;\;\;|A|=d;\;\;\;i,j\in A.
\end{eqnarray}
\end{proposition}
\begin{proof}
Acting with $\Pi(A)$ in Eq.(\ref{cf1}), on a coherent state $\ket{C;k}$ with $k\in A$, we get
\begin{eqnarray}
\Pi(A)\ket{C;k}=\sum _{i,j}G_{ij}(A)\ket{C;i}g_{jk}=\sum _{i}\delta_{ik}\ket{C;i}=\ket{C;k}.
\end{eqnarray}
Also if $\ket {u}$ is a state orthogonal to all $\ket{C;k}$ with $k\in A$ then $\Pi(A)\ket{u}=0$.
Therefore $\Pi(A)$ is a projector to the subspace $H(A)$.
\end{proof}
\begin{remark}
Let $\ket{s}$ be an arbitrary state, $\theta$ a Hermitian operator, and
$s_i=\bra{C;i}s\rangle$.
Using Eq.(\ref{rr}), we express an arbitrary state $\ket{s}$ as
\begin{eqnarray}\label{q3}
\ket{s}=\sum _{i\in A}\sigma _i\ket{C;i};\;\;\;\sigma_i=\sum _{j\in A}G_{ij}(A)s_j ;\;\;\;|A|=d.
\end{eqnarray}
Also using the resolution of the identity in Eq.(\ref{1111}), we express the state $\ket{s}$ as
\begin{eqnarray}\label{q4}
\ket{s}=\frac{1}{d}\sum _{i\in \Omega}s_i\ket{C;i}.
\end{eqnarray}
The expansion in Eq.(\ref{q3}) involves $d$ coherent states labeled with elements of the set $A$.
The expansion in Eq.(\ref{q4}) involves all $d^2$ coherent states. 
We note that
\begin{eqnarray}
\frac{1}{d}\sum _{i\in \Omega}g_{ij}(\Omega)s_j=s_i.
\end{eqnarray}
\end{remark}
Given a Hermitian operator $\theta$, its $Q$-function with respect to the coherent projectors $\Pi(A)$ is given by
\begin{eqnarray}\label{mmm2}
Q(A|\theta)=\frac{1}{d}{\rm Tr}[\Pi(A)\theta];.
\end{eqnarray}
It is easily seen that
\begin{eqnarray}\label{mmm22}
&&\sum _{\alpha,, \beta} Q[A+(\alpha,\beta)|\theta]=|A|{\rm Tr}(\theta)\nonumber\\
&&Q[A+(\alpha,\beta)|\lambda _1\theta _1+\lambda _2\theta _2]=\lambda _1Q[A+(\alpha,\beta)|\theta _1]
+\lambda _2Q[A+(\alpha,\beta)|\theta _2].
\end{eqnarray}

We note that the $d^2$ values of $Q(i|\theta)$ determine the operator $\theta$ (Eq.(\ref{xxxz})), and then the $Q(A|\theta)$:
\begin{eqnarray}\label{QQQ}
Q(A|\theta)=\frac{1}{d}\sum {\rm Tr}[\Pi(A)\theta];\;\;\;\theta _{j}=\sum({\mathfrak M}^{-1})_{ij}Q(i|\theta)
\end{eqnarray}
Here the second equation gives the $\theta$ as $d^2\times 1$ vector which needs to be converted into a $d\times d$ matrix in order to be used in the first equation. 
There are $2^{d^2}-1$ values of $Q(A|\theta)$ (one for each non-empty subset of $\Omega$), but only $d^2$ of them are independent.
Also from Eq.(\ref{123}), it follows that
\begin{eqnarray}
Q(A|\theta)\ne \sum _{i\in A}Q(i|\theta).
\end{eqnarray}

\section{Cooperative game theory and Shapley values}\label{SH}

Given a set $N$ of $|N|$ players, a coalition is a subset $A\subseteq N$. Von Neumann and Morgenstern\cite{S0} introduced the characteristic function 
which is a real valued function that assigns a value $v(A)$ to each subset of players $A\subseteq N$:
\begin{eqnarray}\label{game}
A\;\rightarrow \;v(A);\;\;\;A\subseteq N;\;\;\;v(\emptyset)=0.
\end{eqnarray}
This is interpreted as the `worth' (or `value\rq{}, or `power\rq{}) of the coalition $A$.
$v(N)$ is the total worth, and cooperative game theory aims to find a `sharing rule' of $v(N)$ among the $|N|$ players.
We call $\Sigma _G$ the set of these cooperative games.

There are $2^{|N|}$ subsets of $N$, and consequently the characteristic function takes $2^{|N|}$ values (with $v(\emptyset)=0$). 
Therefore there are $2^{|N|}-1$ degrees of freedom, and the $v(A)$ is a vector in ${\mathbb R}^{{2^{|N|}}-1}$.
If $A=\{i_1,...,i_k\}$ is a coalition of $k$ players, then in general
\begin{eqnarray}\label{149}
v(A)\ne v(i_1)+...+v(i_k).
\end{eqnarray}

Addition of two games is defined with the characteristic function
\begin{eqnarray}
(\lambda _1v_1+\lambda _2v_2)(A)=\lambda _1v_1(A)+\lambda _2v_2(A).
\end{eqnarray}
The marginal contribution of player $i$ to the coalition $A$, is given by
\begin{eqnarray}\label{777AA}
&&{\cal V}(i|A)=v[A\cup\{i\}]-v(A);\;\;\;A\subseteq N\setminus \{i\}\nonumber\\&&{\cal V}(i|\emptyset)=v(i),
\end{eqnarray}
In order to find the share of the player $i$ from the total worth $v(N)$,
we need a method of averaging the marginal contributions ${\cal V}(i|A)$, over all possible coalitions $A$.
Shapley proposed such a method which we describe briefly below\cite{S1,S2,S3}.

Let $\pi$ be a permutation of the players, and ${\cal N}_{\pi}(i)$ the position of $i$ in this permutation.
${\cal A}(\pi; i)$ is the set
\begin{eqnarray}\label{99}
{\cal A}(\pi; i)=\{j\;|\;{\cal N}_{\pi}(j)<{\cal N}_{\pi}(i)\}
\end{eqnarray}
that contains players preceding player $i$ in the permutation $\pi$.
The 
\begin{eqnarray}
{\cal V}[i|{\cal A}(\pi ;i)]=v[{\cal A}(\pi ;i)\cup \{i\}]-v[{\cal A}(\pi;i)]
\end{eqnarray}
quantifies the contribution of player $i$ to the coalition of players $j$ such that ${\cal N}_{\pi}(j)\le {\cal N}_{\pi}(i)$.
In order to treat all players equally, Shapley considered all permutations.
\begin{proposition}\label{pro10}
The Shapley value for player $i$ (i.e., the share of player $i$ from the total worth $v(N)$) is given by the following formulas, which are equivalent to each other.
\begin{itemize}
\item[(1)]
The players join a coalition in an order described by a permutation $\pi$.
The marginal contribution of player $i$, to the coalition ${\cal A}(\pi ;i)$ of players preceding $i$ in the permutation $\pi$, is
${\cal V} [i|{\cal A}(\pi ;i)]$.
Summation is over all $|N|!$ permutations: 
\begin{eqnarray}\label{120}
{\cal S}(i)={\cal S}(i|v)=\frac{1}{|N|!}\sum _{\pi}{\cal V} [i|{\cal A}(\pi ;i)]
\end{eqnarray}
We use two notations ${\cal S}(i)$ and ${\cal S}(i|v)$ because sometimes it is necessary to indicate explicitly which characteristic function we used.
\item[(2)]
We find an average of the marginal contributions ${\cal V}(i|A)$ of player $i$, to all coalitions.
Coalitions with the same cardinality have equal weight.
Coalitions with different cardinalities are equally likely.
\begin{eqnarray}\label{12X}
{\cal S}(i)={\cal S}(i|v)&=&\frac{1}{|N|}\sum _{A\subseteq N\setminus \{i\}}\begin{pmatrix}|N|-1\\|A|\\\end{pmatrix}^{-1}{\cal V}(i|A)\nonumber\\
&=&\frac{1}{|N|}\left [v(i)+\frac{1}{|N|-1}\sum _{|A|=1}{\cal V}(i|A)+\begin{pmatrix}|N|-1\\2\\\end{pmatrix}^{-1}\sum _{|A|=2}{\cal V}(i|A)+...
\right].
\end{eqnarray}
There are $\begin{pmatrix}|N|-1\\|A|\\\end{pmatrix}$ coalitions in $N\setminus \{i\}$ with a given cardinality $|A|$, 
and this leads to the coefficient $\begin{pmatrix}|N|-1\\|A|\\\end{pmatrix}^{-1}$. 
\end{itemize} 
\end{proposition}
\begin{proof}
The equivalence of Eq.(\ref{12X}) to Eq.(\ref{120}) is seen from the fact that
for a given $i$ and a given subset $A\subseteq N\setminus \{i\}$, there are 
$|A|!$ permutations in which the labels in $A$ are before $i$, and $(|N|-|A|-1)!$ permutations in which the 
rest $|N|-|A|-1$ labels (i.e., the labels in $N \setminus (A\cup \{i\})$) are after $i$.
Therefore there is a total of $(|N|-|A|-1)!|A|!$ permutations, in which the labels in $A$ are before $i$, and the rest of the labels are after $i$.
This gives the result in Eq.(\ref{12X}).
\end{proof}
\begin{remark}\label{278}
There are many characteristic functions $v(A)$, which lead to the same Shapley values ${\cal S}(i)$.
Indeed, the characteristic function $v(A)$, and the Shapley values ${\cal S}(i)$ 
can be viewed as vectors in ${\mathbb R}^{2^{|N|}-1}$ and ${\mathbb R}^{|N|}$, correspondingly.
For a given ${\cal S}(i)$, Eqs(\ref{12X}) form a system of $|N|$ equations with $2^{|N|}-1$ unknowns.
In section \ref{equiv}, we will introduce equivalence classes of characteristic functions which have the same Shapley values.
\end{remark}

Below we prove three properties of the Shapley values.
Shapley introduced these properties as axioms\cite{S1,S2,S3}, and proved that they lead uniquely to Eqs.(\ref{120}),(\ref{12X}).
\begin{proposition}\label{mnb}
\mbox{}
\begin{itemize}
\item[(1)]
\begin{eqnarray}\label{bbb}
\sum _i{\cal S}(i)=v(N).
\end{eqnarray}
This reflects the fact that the Shapley values are the share of each player from the total value $v(N)$.
We call this equation `resolution of $v(n)$ into parts for the various players\rq{},
because this will help later to make the analogy with the resolution of the identity in terms of the dressed coherent states. 
\item[(2)]
Let $\pi$ be a permutation of the players which maps the player $i$ into $\pi (i)$, and more generaly the subset $A\subseteq N$
into $\pi(A)$. We define permutation transformations that map a characteristic function $v$ into the characteristic function $v_{\pi}$, such that
\begin{eqnarray}\label{perm}
v_{\pi}[\pi(A)]=v(A).
\end{eqnarray}
Here we `jumble\rq{} the  $2^{|N|}-1$ values of $v(A)$, and assign the value of $v(A)$ to the set
$\pi(A)$, which has the same cardinality as $A$.
Then  
\begin{eqnarray}\label{bbbc}
{\cal S}[\pi(i)|v_{\pi}]={\cal S}(i|v).
\end{eqnarray}
\item[(3)]
If $v_1,v_2$ are characteristic functions, then $\lambda _1v_1+\lambda _2 v_2$ is also a characteristic function and
\begin{eqnarray}\label{add}
{\cal S}(i|\lambda _1v_1+\lambda _2 v_2)=\lambda _1{\cal S}(i|v_1)+\lambda _2 {\cal S}(i|v_2);\;\;\;\lambda _1, \lambda _2 \in {\mathbb R}.
\end{eqnarray}

\end{itemize}
\end{proposition}
\begin{proof}
\mbox{}
\begin{itemize}
\item[(1)]
In order to prove that $\sum _i{\cal S}(i)=v(N)$ we point out that  
\begin{eqnarray}
|N|\sum _i {\cal S}(i)=\sum _i\sum _{A\subseteq N\setminus \{i\}}\begin{pmatrix}|N|-1\\|A|\\\end{pmatrix}^{-1}[v[A\cup\{i\}]-v(A)]
\end{eqnarray}
In this sum the term $-\begin{pmatrix}|N|-1\\|A|\\\end{pmatrix}^{-1}v(A)$ with a fixed set $A \subset N$, appears $|N|-|A|$ times (for all $i\in N\setminus A$).
They cancel the terms  $\begin{pmatrix}|N|-1\\|B|\\\end{pmatrix}^{-1}v(B\cup \{i\})$ where $B=A\setminus \{i\}$ for all $i\in A$ (there are $|A|$ such sets).
Indeed, we get
\begin{eqnarray}
-\begin{pmatrix}|N|-1\\|A|\\\end{pmatrix}^{-1}(|N|-|A|)v(A)+\begin{pmatrix}|N|-1\\|A|-1\\\end{pmatrix}^{-1}|A|v(A)=0
\end{eqnarray}
There is no such cancellation for $A=N$, and these terms give $\sum _i{\cal S}(i)=v(N)$.
\item[(2)]
\begin{eqnarray}\label{777}
{\cal S}[\pi(i)|v_{\pi}]&=&\frac{1}{|N|}\sum _{A\subseteq N\setminus \{\pi(i)\}}\begin{pmatrix}|N|-1\\|A|\\\end{pmatrix}^{-1}[v_{\pi}(A\cup \{\pi(i)\})-v_{\pi}(A)]
\end{eqnarray}
We put $A=\pi(B)$ where $B\subseteq N\setminus \{i\}$. Then $|A|=|\pi(B)|$ and we get
\begin{eqnarray}
{\cal S}[\pi(i)|v_{\pi}]&=&\frac{1}{|N|}\sum _{B\subseteq N\setminus \{i\}}\begin{pmatrix}|N|-1\\|B|\\\end{pmatrix}^{-1}\{v_{\pi}[\pi(B)\cup \{\pi(i)\}]-v_{\pi}[{\pi}(B)]\}\nonumber\\&=&
\frac{1}{|N|}\sum _{B\subseteq N\setminus \{i\}}\begin{pmatrix}|N|-1\\|B|\\\end{pmatrix}^{-1}[v(B\cup \{i\})-v(B)]={\cal S}(i|v).
\end{eqnarray}

\item[(3)]
The Shapley value is a linear function of the $v(A)$, and this proves Eq.(\ref{add}).

\end{itemize}
\end{proof}

\subsection{M\"obius transforms}
M\"obius transforms have been introduced by Rota in combinatorics\cite{R,R1}, as a generalization of the 
`inclusion-exclusion\rq{} principle in set theory, which gives the cardinality of the union of overlapping sets.
It is a method for finding the overlaps, and avoiding the double-counting.
Rota generalized this to partially ordered structures.
The M\"obius transform of the characteristic function $v(A)$ and its inverse, are defined as:
\begin{eqnarray}\label{m14}
&&{\mathfrak d} (B)=\sum _{A\subseteq B} (-1)^{|A|-|B|}v(A);\;\;\;\;A,B\subseteq N\nonumber\\
&&v (A)=\sum _{B\subseteq A}{\mathfrak d} (B).
\end{eqnarray}
They quantify the added value (which might be positive or negative) in the coalitions.
For example:
\begin{eqnarray}\label{m12}
&&{\mathfrak d} (i_1)=v(i_1);\;\;\;{\mathfrak d} (i_1,i_2)=v(i_1,i_2)-v(i_1)-v(i_2)\nonumber\\
&&{\mathfrak d} (i_1,i_2,i_3)=v(i_1,i_2,i_3)-v(i_1,i_2)-v(i_1,i_3)-v(i_2,i_3)+v(i_1)+v(i_2)+v(i_3).
\end{eqnarray}
If for every coalition $A=\{i_1,...,i_k\}$ the $v(A)= v(i_1)+...+v(i_k)$ (i.e., there is no added value), 
then  all the ${\mathfrak d} (B)$ with $|B|\ge 2$, are zero.

\begin{proposition}\label{basis}
\mbox{}
\begin{itemize}
\item[(1)]
Let ${\mathfrak V}_A$ be the characteristic functions defined by
\begin{eqnarray}\label{444}
&&{\mathfrak V}_A(B)=1\;\;{\rm if}\;\; A\subseteq B\nonumber\\
&&{\mathfrak V}_A(B)=0\;\;{\rm otherwise}.
\end{eqnarray}
The set $\{{\mathfrak V}_A|A\subseteq N\}$ is a basis, in the sense that any
characteristic function $v$ can be written as 
\begin{eqnarray}\label{basis}
v (A)=\sum _{C\subseteq N}{\mathfrak d} (C){\mathfrak V}_C(A).
\end{eqnarray}
The coefficients are the M\"obius transforms of $v(A)$ (Eq.(\ref{m14}))
\item[(2)]
The Shapley values of ${\mathfrak V}_A$  are
\begin{eqnarray}\label{4fg}
&&{\cal S}(i|{\mathfrak V}_A)=\frac{1}{|A|}\;\;{\rm if}\;\; i\in A\nonumber\\
&&{\cal S}(i|{\mathfrak V}_A)=0\;\;{\rm if}\;\; i\notin A.
\end{eqnarray}
Therefore the Shapley values of any characteristic function $v$ are:
\begin{eqnarray}\label{1290}
{\cal S}(i)=\sum _{A\ni i}\frac{{\mathfrak d}(A)}{|A|}.
\end{eqnarray}
The summation involves subsets $A$ of $N$, which contain $i$.
The first term is ${\mathfrak d}(i)=v(i)$, and the rest are the added values ${\mathfrak d}(A)$
in all the coalitions that $i$ participates. They are divided by $|A|$ because the added value 
in each coalition, is divided equally among its members.
\end{itemize}
\end{proposition}
\begin{proof}
\mbox{}
\begin{itemize}

\item[(1)]
We rewrite the inverse Mobius transform in Eq.(\ref{m14}), as
\begin{eqnarray}
v (A)=\sum _{B\subseteq A}{\mathfrak d} (B)=\sum _{C\subseteq N}{\mathfrak d} (C){\mathfrak V}_C(A).
\end{eqnarray}
\item[(2)]
Eq.(\ref{4fg}) has been proved by Shapley \cite{S1,S2}, using the properties in Eqs(\ref{bbb}),(\ref{bbbc}).
The Shapley value is a linear function of the characteristic function (Eq.(\ref{add})) and therefore
\begin{eqnarray}
{\cal S}(i|v)=\sum _{C\subseteq N}{\mathfrak d} (C){\cal S}[i|{\mathfrak V}_C(A)]=\sum _{C\ni i}\frac{{\mathfrak d} (C)}{|C|}.
\end{eqnarray}
\end{itemize}
\end{proof}
\begin{cor}
If for every coalition $A=\{i_1,...,i_k\}$ the $v(A)= v(i_1)+...+v(i_k)$, then ${\cal S}(i)=v(i)$.
\end{cor}
\begin{proof}
The $v(A)= v(i_1)+...+v(i_k)$, implies that all the ${\mathfrak d} (B)$ with $|B|\ge 2$, are zero. Then Eq.(\ref{1290}) proves that ${\cal S}(i)=v(i)$.
\end{proof}
Physically, the additivity relation $v(A)= v(i_1)+...+v(i_k)$ means that there is no added value in the coalitions, and we get
\begin{eqnarray}
{\cal S}(i)={\cal V}(i|A)=v(i).
\end{eqnarray}
The analogue of this in a quantum context, would be to have an orthogonal basis so that the
projector corresponding to an aggregation of vectors in the basis, is equal to the sum of the projectors for the various vectors.

\begin{remark}
Our aim is to transfer these ideas in a quantum context, and for this reason we introduce some `quantum terminology'.
We call $v(i)$ the `bare' worth of the player $i$, which through the interaction with the other players is 
`dressed' or `renormalized' into the Shapley 
value $S(i)$. This is analogous to `bare electrons' in quantum field theory, which are `dressed' or `renormalized'
through the interaction with photons.
\end{remark}
\subsection{Example}\label{100}
In a factory workers $1,2,3$ working individually, produce $1,0,2$ items of the same product (per hour), correspondingly. 
Worker $2$ has a very special expertise and cannot produce a full item on his own, but he can contribute in collaborations.
The collaboration of $(1,2)$ produces $4$  items, 
the collaboration of $(1,3)$ produces $3$  items,
and the collaboration of $(2,3)$ produces $6$  items.
Furthermore the collaboration of $(1,2,3)$ produces $8$  items.
In this case
\begin{eqnarray}
&&v(\emptyset)=0;\;\;\;v(1)=1;\;\;\;v(2)=0;\;\;\;v(3)=2\nonumber\\
&&v(1,2)=4;\;\;\;v(1,3)=3;\;\;\;v(2,3)=6;\;\;\;v(1,2,3)=8.
\end{eqnarray}
The contribution of various players to various coalitions is
\begin{eqnarray}
&&{\cal V}(1|\emptyset)=1;\;\;\;{\cal V}(1|2)=4;\;\;\;{\cal V}(1|3)=1;\;\;\;{\cal V}(1|2,3)=2\nonumber\\
&&{\cal V}(2|\emptyset)=0;\;\;\;{\cal V}(2|1)=3;\;\;\;{\cal V}(2|3)=4;\;\;\;{\cal V}(2|1,3)=5\nonumber\\
&&{\cal V}(3|\emptyset)=2;\;\;\;{\cal V}(3|1)=2;\;\;\;{\cal V}(3|2)=6;\;\;\;{\cal V}(3|1,2)=4.
\end{eqnarray}
We calculate the Shapley values using the three formulas in Eqs.(\ref{120}),(\ref{12X}),(\ref{1290}).

In order to use Eq.(\ref{120}), we consider the $6$ permutations of $1,2,3$ and we get
\begin{eqnarray}\label{353}
{\cal S}(1)=\frac{1}{6}\left[v(1)+v(1)+{\cal V}(1|2)+{\cal V}(1|2,3)+{\cal V}(1|3)+{\cal V}(1|3,2)\right ]=\frac{11}{6},
\end{eqnarray}
and also
\begin{eqnarray}\label{353a}
{\cal S}(2)=\frac{17}{6};\;\;\;{\cal S}(3)=\frac{10}{3}.
\end{eqnarray}

We also use Eq.(\ref{12X}), and we get 
\begin{eqnarray}\label{353}
{\cal S}(1)=\frac{1}{3}\left[1+\frac{1}{2}(4+1)+2\right ]=\frac{11}{6}.
\end{eqnarray}
In a similar way we get the values for $S(2), S(3)$ given in Eq.(\ref{353a}).

We also calculate the Shapley values using Eq.(\ref{1290}).
The M\"obius transform of the characteristic function is
\begin{eqnarray}
&&{\mathfrak d}(1)=1;\;\;\;{\mathfrak d}(2)=0;\;\;\;{\mathfrak d}(3)=2\nonumber\\
&&{\mathfrak d}(1,2)=3;\;\;\;{\mathfrak d}(1,3)=0;\;\;\;{\mathfrak d}(2,3)=4;\;\;\;{\mathfrak d}(1,2,3)=-2,
\end{eqnarray}
and it gives
\begin{eqnarray}
{\cal S}(1)={\mathfrak d}(1)+\frac{1}{2}[{\mathfrak d}(1,2)+{\mathfrak d}(1,3)]+\frac{1}{3}{\mathfrak d}(1,2,3)=\frac{11}{6}.
\end{eqnarray}
Here the first term is ${\mathfrak d}(1)=v(1)$.
The second term is the added values in the coalitions $(1,2)$ and $(1,3)$ divided by $2$, because each of them has two players.
The third term is the added value in the coalition $(1,2,3)$ (which is negative) divided by $3$, because it has three players.
In a similar way we get the values for $S(2), S(3)$ given in Eq.(\ref{353a}).

It is seen that the `bare' worth of worker $1$ is $v(1)=1$, and through the interaction with the other players it is 
`dressed' or `renormalized' into $S(1)=11/6$.
Similar `dressing' occurs for the other players.

\section{Analogies between cooperative game theory and aggregations of coherent states}\label{IV}

In this section we establish a precise correspondence between cooperative game theory and aggregations of coherent states.
In order to do this, we first discuss  some quantities in a quantum context which are analogous to those in cooperative game theory.
They provide mathematical and physical insight to the non-orthogonality and non-commutativity of coherent projectors,
by studying not only the individual role of coherent states in the formalism, but also the role of aggregations of coherent states.

\subsection{M\"obius operators: quantifying the non-orthogonality and non-commutativity of coherent projectors }\label{IVA}

We have discussed M\"obius transforms in a different context in refs\cite{v16a,v16b},
and here we give briefly the relevant formulas. 
The M\"obius transform of the coherent projectors $\Pi(A)$, is given by:
\begin{eqnarray}\label{m11}
{\mathfrak D} (B)=\sum _{A\subseteq B} (-1)^{|A|-|B|}\Pi(A);\;\;\;\;A,B\subseteq\Omega.
\end{eqnarray}
Some examples are:
\begin{eqnarray}\label{m12}
&&{\mathfrak D} (1)=\Pi(1);\;\;\;{\mathfrak D} (1,2)=\Pi(1,2)-\Pi(1)-\Pi(2)\nonumber\\
&&{\mathfrak D} (1,2,3)=\Pi(1,2,3)-\Pi(1,2)-\Pi(1,3)-\Pi(2,3)+\Pi(1)+\Pi(2)+\Pi(3).
\end{eqnarray}
The inverse M\"obius transform is
\begin{eqnarray}\label{m13}
\Pi (A)=\sum _{B\subseteq A}{\mathfrak D} (B).
\end{eqnarray}
We note that if instead of coherent states we use an orthonormal set of states,
then the relation in Eq.(\ref{123}) becomes equality, and all the ${\mathfrak D} (B)$ with $|B|\ge 2$, are zero.

We have explained in \cite{v16a,v16b} that:
\begin{itemize}
\item
M\"obius transforms give the `added value' in an aggregation of coherent states, and they are intimately related 
to the non-orthogonal nature of coherent states.
For example, the ${\mathfrak D} (1,2)$ shows the difference of the projector $\Pi(1,2)$ 
of the aggregation of the coherent states labelled with $1,2$, from the sum $\Pi(1)+\Pi(2)$.
For orthogonal projectors ${\mathfrak D} (1,2)=0$.
\item
M\"obius transforms are intimately related to commutators that involve coherent projectors, e.g.,
\begin{eqnarray}\label{375}
&&[\Pi(i),\Pi(j)]={\mathfrak D} (i,j)[\Pi(i)-\Pi(j)]\nonumber\\
&&[[\Pi(i), \Pi(k)], \Pi(j)]=\Pi(j){\mathfrak D} (i,j,k)[\Pi(i)-\Pi(k)]+[\Pi(i)-\Pi(k)]{\mathfrak D} (i,j,k)\Pi(j).
\end{eqnarray}
Therefore, using M\"obius transforms, is equivalent to taking into account non-commutativity effects.
\end{itemize}
These comments are general results \cite{v16a} for projectors to any subspaces of the Hilbert space.
In most of the paper we are interested in the coherent projectors, but in section \ref{gen} we will also use the more general case.

\subsection{Contribution of a coherent state to an aggregation}\label{IVB}

Let $A \subseteq \Omega$, and $i\in \Omega \setminus A$. The projector (of rank one) 
\begin{eqnarray}\label{120n}
\varpi (i|A)=\Pi(A\cup \{i\})-\Pi(A)\ne \Pi(i);\;\;\;\;\varpi (i|\emptyset)=\Pi(i),
\end{eqnarray}
quantifies the contribution of the coherent state $\ket{C;i}$ to the aggregation of coherent states labeled with $A\cup \{i\}$.
Using the Gram-Schmidt orthogonalization algorithm, we can express $\varpi (i|A)$ as in Eq.(\ref{gg}):
\begin{eqnarray}\label{ggc}
&&\varpi(i|A)=
\frac{\Pi^{\perp}(A)\Pi(i)\Pi^{\perp}(A)}{{\rm Tr}[\Pi^{\perp}(A)\Pi(i)]};\;\;\;\Pi^{\perp}(A)={\bf 1}-\Pi (A);\;\;\;i\in \Omega \setminus A.
\end{eqnarray}
Examples of this, are: 
\begin{eqnarray}\label{gg}
&&\Pi (1,2)=\Pi(1)+\varpi(2|1);\;\;\;
\varpi(2|1)=\frac{\Pi^{\perp}(1)\Pi(2)\Pi^{\perp}(1)}{{\rm Tr}[\Pi^{\perp}(1)\Pi(2)]}\nonumber\\
&&\Pi (1,2,3)=\Pi(1,2)+\varpi(3|1,2);\;\;\;
\varpi(3|1,2)=\frac{\Pi^{\perp}(1,2)\Pi(3)\Pi^{\perp}(1,2)}{{\rm Tr}[\Pi^{\perp}(1,2)\Pi(3)]}
\end{eqnarray}
etc.

Using Eqs(\ref{45}),(\ref{44}), we can prove\cite{v16a} the following coherence properties for $\varpi [(\alpha, \beta)|A]$
with fixed set $A$ (we use here the pair of indices notation):
\begin{eqnarray}\label{58}
&&\frac{1}{d}\sum _{\gamma, \delta}\varpi [(\alpha+\gamma, \beta +\delta)|A+(\gamma, \delta)]={\bf 1};\;\;\;(\alpha ,\beta)\in \Omega \setminus A,
\end{eqnarray}
and
\begin{eqnarray}\label{27}
&&D(\gamma, \delta)\varpi [(\alpha, \beta)|A]D^{\dagger}(\gamma, \delta)=
\varpi [(\alpha+\gamma, \beta+\delta)|A+(\gamma, \delta)]\nonumber\\
&&A\subseteq \Omega \setminus \{(\alpha ,\beta)\};\;\;\;\;A+(\gamma, \delta)\subseteq \Omega \setminus \{(\alpha +\gamma,\beta +\delta)\}.
\end{eqnarray}
Also since we use generic fiducial vectors, in the case $|A|\ge d$ we get $\Pi(A)={\bf 1}$, and therefore
\begin{eqnarray}\label{120a}
|A|\ge d\;\;\rightarrow\;\;\varpi (i|A)=0.
\end{eqnarray}
Eqs.(\ref{58}),(\ref{27}) show that the $\varpi [(\alpha, \beta)|A]$ are coherent density matrices, according to the definition \ref{def}.
In fact they are generalized coherent projectors of rank one, in the sense that for $A=\emptyset$ they reduce to the coherent projectors in section \ref{vv}.

If $A,B$ are two subsets of $\Omega$, and $i\in \Omega \setminus (A\cup B)$, then
\begin{eqnarray}\label{120b}
\varpi (i|A)\ne \varpi (i|B).
\end{eqnarray}
The contribution of $i$ to the aggregation $A\cup \{i\}$, is different from 
the contribution of $i$ to the aggregation $B\cup \{i\}$.

We can also define $Q$ functions with respect to the $\varpi (i|A)$ as:
\begin{eqnarray}
 q_{\theta}(i|A)=\frac{1}{d} {\rm Tr}[\theta \varpi (i|A)]=Q(A\cup \{i\}|\theta)-Q(A|\theta);\;\;\;i\in \Omega \setminus A.
\end{eqnarray}
Eq.(\ref{58}) implies that
\begin{eqnarray}
&&\frac{1}{d}\sum _{\gamma, \delta}q_{\theta}[(\alpha+\gamma, \beta +\delta)|A+(\gamma, \delta)]={\rm Tr}(\theta);\;\;\;(\alpha ,\beta)\in \Omega \setminus A.
\end{eqnarray}

\subsection{Correspondence between cooperative game theory and aggregations of coherent states}
The correspondence between quantities in cooperative game theory, and quantities in a quantum context, is as follows.
We introduce a bijective map between the set $\Omega$ of $d^2$ labels for coherent states, and the set $N$ of $|N|=d^2$ players:
\begin{eqnarray}
\Omega\;\rightarrow\;N;\;\;\;|N|=d^2.
\end{eqnarray}
For simplicity, we use the same notation for a subset $A$ of $\Omega$, and the corresponding subset of $N$.
For a Hermitian operator $\theta$, we call
`corresponding game' the one  with characteristic value 
\begin{eqnarray}\label{12C}
v(A)=dQ(A|\theta)={\rm Tr}[\theta \Pi (A)],
\end{eqnarray}
for all subsets $A$.
Different Hermitian operators $\theta $ correspond to different cooperative games $v$.

$dQ(i|\theta)$ can be viewed as the `worth\rq{} of the coherent state $\ket{C;i}$ in the description of the Hermitian operator $\theta$, analogous to the worth $v(i)$ of player $i$ in cooperative game theory. 
For example, if $\theta $ is a density matrix, a measurement with the projector $\Pi(i)=\ket{C;i}\bra{C;i}$ will give `yes\rq{} with probability
${\rm Tr}[\theta \Pi (i)]=dQ(i|\theta)$.
More generally $Q(A|\theta)$ can be viewed as the `worth\rq{} of the aggregation $A$ of coherent states for the description of the operator $\theta$, analogous to 
the worth $v(A)$ of players in the coalition $A$, in the game $v$.

If $\theta _1, \theta _2$ are Hermitian operators and $v_1,v_2$ the corresponding cooperative games
(i.e., $dQ(A|\theta _1)=v_1(A)$ and $dQ(A|\theta _2)=v_2(A)$ for all subsets $A$), then it is easily seen that
\begin{eqnarray}
dQ(A|\lambda _1\theta _1+\lambda _2 \theta _2)=[\lambda _1v_1+\lambda _2v_2](A);\;\;\;\lambda _1, \lambda _2\in {\mathbb R}.
\end{eqnarray}

\subsection{Embedding of the set ${\Sigma _Q}$ of physical $Q$-functions into the larger set ${\Sigma _{VQ}}$ of virtual $Q$-functions}

Sometimes it is helpful to work in a larger set than the `physical set' of the problem, and at the end 
of the calculation, we can restrict the results into the physical set. This is the case if for the larger set, there are many known results 
in a different context, which are readily available (in our case in the context of cooperative game theory). 

The map of Eq.(\ref{12C}) is an `into function'.
There are many characteristic functions with no analogous $Q$-function.
We have explained earlier (Eq.(\ref{QQQ})) that the $d^2$ values of $Q(i|\theta)$ determine uniquely the $Q(A|\theta)$.
In contrast, the $|N|=d^2$ values of $v(i|\theta)$ do not define  the $v(A|\theta)$.
The $Q$ function of a Hermitian operator has $d^2$ degrees of freedom, while the 
characteristic function has $2^{|N|}-1=2^{d^2}-1$ degrees of freedom.

In order to transfer results from cooperative game theory, in the context of aggregations of coherent states, 
it is helpful to have a bijective map between the set of $Q$-functions and the set of characteristic functions.
For this reason we enlarge the set ${\Sigma _Q}$ of `physical' $Q$-functions, 
into the ${\Sigma}_{VQ}$ of `virtual' $Q$-functions. They are functions that assign a real value $Q(A)$ to each aggregation of coherent states described 
by the subset $A\subseteq \Omega$:
\begin{eqnarray}\label{game1}
A\;\rightarrow \;Q(A);\;\;\;Q(\emptyset)=0.
\end{eqnarray}
The term `virtual' means that  we discard the restriction that the $d^2$ values of $Q(i|\theta)$ define uniquely
the $Q(A|\theta)$.    
The $Q$-functions in the set ${\Sigma }_{VQ}\setminus {\Sigma _Q}$ are `non-physical', and they have no 
corresponding operator $\theta$ (for this reason we omit $\theta$ in the notation $Q(A)$).
The definition in Eq.(\ref{game1}) is analogous to the definition of the characteristic function in 
Eq.(\ref{game}). 

It is clear that we have a bijective map from the set of cooperative games $\Sigma _G$ to ${\Sigma _{VQ}}$,
with the $v(A)$ corresponding to $dQ(A)$.
The sets $\Sigma _G$ and ${\Sigma _{VQ}}$ are isomorphic ($\Sigma _G \simeq {\Sigma _{VQ}}$).
Therefore we can transfer results from cooperative game theory to the set ${\Sigma _{VQ}}$ of virtual $Q$ functions, and then restrict ourselves to the physical $Q$ functions.

\begin{remark}
If $\theta$ is a density matrix, all values of $Q(A|\theta)$ can be measured in yes/no experiments with the projectors $\Pi(A)$.
The $\Pi(A)$ do not commute and different ensembles describing the same density matrix $\theta$ need to be used.
Ideal measurements will satisfy Eq.(\ref{QQQ}), and will belong in ${\Sigma _{Q}}$, but real (noisy) measurements will 
belong in ${\Sigma _{VQ}}$. 
In the latter case Eq.(\ref{QQQ}) is an incompatible system of more than $d$ equations, with $d$ unknowns (the $\theta_j)$.
Computer libraries, have programs which give solutions to such systems, with minimum error.
Such a solution is probably more accurate than the one obtained from measurements of the $d$ values of $Q(i|\theta)$ only.
The redundancy in having all the values of $Q(A|\theta)$ is desirable, because it can lead to error correction.
In this sense, the space ${\Sigma _{VQ}}$ can be useful in the study of noisy measurements of $Q$-functions, but we do not pursue this direction. 
\end{remark}

\subsection{Permutation transformations in ${\Sigma _{VQ}}$}

We consider permutation transformations analogous to those in Eq.(\ref{perm}).
Let $\pi$ be a permutation of the $d^2$ elements in $\Omega$. 
We define the function $Q_{\pi}$ as
\begin{eqnarray}\label{2c7}
Q_{\pi}[\pi(A)]=Q(A).
\end{eqnarray}
Here we `jumble\rq{} the  $2^{d^2}-1$ values of $Q(A)$,
taking into account the constraint that the sets $A$ and $\pi(A)$ have the same cardinality.

\begin{proposition}\label{rfv}
For a physical $Q$-function $Q(A|\theta)$ in ${\Sigma _{Q}}$, the corresponding $Q_{\pi}$-function will (in general) be  a
virtual $Q$-function in ${\Sigma _{VQ}}$.
\end{proposition}
\begin{proof}
We consider 
the $Q$-function $Q(i|\theta)$ of a Hermitian operator $\theta$.
We define an operator $\theta _{\pi}$, such that 
\begin{eqnarray}
Q[\pi(i)|\theta _{\pi}]=Q(i|\theta).
\end{eqnarray}
For a given $Q(i|\theta)$, Eq.(\ref{2c7}) with $|A|=1$ is a system of $d^2$ equations with $d^2$ unknowns (the matrix elements of $\theta _{\pi}$).
So we can calculate $\theta _{\pi}$, but then in general 
\begin{eqnarray}\label{2c8}
Q[\pi(A)|\theta _{\pi}]\ne Q(A|\theta).
\end{eqnarray}
To get equality in this equation with given $Q(A|\theta)$, requires to satisfy a system of $2^{d^2}-1$ equations 
(one for each non-empty subset $A$ of $\Omega$), with $d^2$ unknowns  (the matrix elements of $\theta _{\pi}$).
Therefore Eq.(\ref{2c8}) is not an equality, in general. It follows that in Eq.(\ref{2c7}), the $Q_{\pi}$-function will (in general) be  a
virtual $Q$-function. 

Only special permutations (e.g., $(\alpha, \beta)\rightarrow (\alpha+1, \beta)$) lead to automorphisms of the Heisenberg-Weyl group, and to
physical $Q_{\pi}$-functions.
\end{proof}

\subsection{Restrictions on $\Sigma _G \simeq {\Sigma _{VQ}}$}

The following remarks reduce the size of the sets $\Sigma _G \simeq {\Sigma _{VQ}}$, and this can be helpful in practical calculations:
\begin{itemize}

\item[(1)]
For $|A|\ge d$ we get $dQ(A)={\rm Tr}(\theta)$ (Eq.(\ref{g7})). Therefore
we need to consider cooperative games with characteristic functions such that 
\begin{eqnarray}\label{wq}
|A|\ge d\;\rightarrow\;v(A)={\rm Tr}(\theta).
\end{eqnarray}
The number of degrees of degrees of freedom in these games is reduced from $2^{|N|}-1=2^{d^2}-1$ to
\begin{eqnarray}\label{wq1}
{\mathfrak n}_d=\begin{pmatrix}d^2\\1\\\end{pmatrix}+\begin{pmatrix}d^2\\2\\\end{pmatrix}+...+\begin{pmatrix}d^2\\d-1\\\end{pmatrix}
\end{eqnarray}
\item[(2)]
The characteristic function $v(A)$ is called superadditive, if
\begin{eqnarray}\label{28}
A_1\cap A_2=\emptyset\;\rightarrow \;v(A_1)+v(A_2)\le v(A_1\cup A_2).
\end{eqnarray}
This means that the `worth' of a coalition is at least equal to its parts acting separately.
The analogue of Eq.(\ref{28}) in a quantum context, is not valid. Indeed if
$A_1\cap A_2=\emptyset$ the $\Pi(A_1\cup A_2)-\Pi(A_1)-\Pi(A_2)$ is not a positive semi-definite operator.
Assuming that $|A_1 \cup A_2|\le d$, we get
\begin{eqnarray}
{\rm Tr}[\Pi(A_1\cup A_2)-\Pi(A_1)-\Pi(A_2)]=(|A_1|+|A_2|)-|A_1|-|A_2|=0.
\end{eqnarray}
which shows that it cannot have all its eigenvalues non-negative.
Consequently $Q(A_1|\theta)+Q(A_2|\theta)$ might be smaller or bigger than $Q(A_1\cup A_2|\theta)$.
For this reason, we consider cooperative games with characteristic function which is not superadditive.
We mention this because there are many results in the cooperative games literature,
which assume superadditive characteristic functions, and which are 
not relevant here.
\item[(3)]
The characteristic function $v(A)$ is called monotonic, if
\begin{eqnarray}\label{27A}
A_1\subseteq A_2\;\rightarrow \;v(A_1)\le v(A_2).
\end{eqnarray}
This means that adding players to a coalition, cannot `harm' the coalition.
Monotonic characteristic functions, take non-negative values 
\begin{eqnarray}
v(N)\ge v(A)\ge v(\emptyset)=0.
\end{eqnarray}
For a positive semi-definite operators $\theta$ we get $Q(i|\theta)\ge 0$, and therefore 
the characteristic function of the corresponding cooperative game,
 should obey the monotonicity property in Eq.(\ref{27A}).
\end{itemize}
It is seen that coherence in finite quantum systems is analogous to cooperative game theory with characteristic function that obeys Eq.(\ref{wq}), 
and does not obey the superadditivity property in Eq.(\ref{28}).
We note that for positive semi-definite operators $\theta$ (e.g., density matrices),
the characteristic function of the corresponding cooperative game, should obey the monotonicity property in Eq.(\ref{27A}).

With these remarks, we have reduced the size of the sets $\Sigma _G \simeq {\Sigma _{VQ}}$, 
and this can be helpful in practical calculations.

\section{ Shapley values as dressed $Q$-functions }\label{V}

In this section, we define Shapley values $S(i)$ for the $Q$-functions.
We first define briefly the quantum analogues of some quantities which we used in cooperative games.

The analogue of ${\cal V}(i|A)$ in Eq.(\ref{777AA}), is
\begin{eqnarray}
q(i|A)=Q(A\cup \{i\})-Q(A);\;\;\;i\in \Omega;\;\;A\subseteq \Omega -\{i\}.
\end{eqnarray}
The M\"obius transform of the $Q$-functions are defined in analogous way to Eqs(\ref{m14}):
\begin{eqnarray}\label{m15}
\Delta(B)=\sum _{A\subseteq B} (-1)^{|A|-|B|}Q(A);\;\;\;Q (A)=\sum _{B\subseteq A}\Delta (B).
\end{eqnarray}
In the case of $Q$-functions in ${\Sigma _{Q}}\subset {\Sigma _{VQ}}$, there exist operator $\theta$ such that
\begin{eqnarray}
&&q(i|A)=Q(A\cup \{i\}|\theta)-Q(A|\theta)=\frac{1}{d} {\rm Tr}[\theta \varpi (i|A)]\nonumber\\
&&\Delta(B)=\frac{1}{d}\sum _{A\subseteq B} (-1)^{|A|-|B|}{\rm Tr}[\theta \Pi (A)]=\frac{1}{d}{\rm Tr}[\theta {\mathfrak D} (B)]
\end{eqnarray}

We also define in ${\Sigma}_{VQ}$, the analogue of the characteristic functions 
in Eq.(\ref{444}), as
\begin{eqnarray}\label{97}
&&d{\mathfrak Q}_A(B)=1\;\;{\rm if}\;\; A\subseteq B\nonumber\\
&&d{\mathfrak Q}_A(B)=0\;\;{\rm otherwise}.
\end{eqnarray}
They are virtual $Q$-functions, and they are a basis in ${\Sigma}_{VQ}$.
Indeed, in analogy to Eq.(\ref{basis}), we can express a $Q$-function in ${\Sigma}_{VQ}$ as
\begin{eqnarray}
Q(A)=\sum _{C\subseteq \Omega}\Delta (C){\mathfrak Q}_C(A).
\end{eqnarray}

We have established an isomorphism between cooperative game theory, and $Q$-functions in ${\Sigma}_{VQ}$. We now introduce
the Shapley values $S(i)$ for the $Q$-functions, by one of the following formulas which we have shown 
(in the context of cooperative game theory) to be equal to each other.
\begin{definition}
The Shapley values $S(i)$, for the $Q$-functions in $\Sigma _{VQ}$, are given by 
\begin{eqnarray}\label{99}
S(i)=\frac{d}{d^2!}\sum _{\pi}q [i|{\cal A}(\pi ;i)]=\frac{1}{d}\sum _{A\subseteq \Omega \setminus \{i\}}\begin{pmatrix}d^2-1\\|A|\\\end{pmatrix}^{-1}q(i|A)
=d\sum _{A\ni i}\frac{\Delta (A)}{|A|};\;\;\;i\in \Omega.
\end{eqnarray}
In the case of $Q$-functions in the subset ${\Sigma _{Q}}$ of ${\Sigma _{VQ}}$, they can be expressed in terms of an operator $\theta$ as
\begin{eqnarray}\label{6nm}
S(i)=\frac{1}{(d^2!)}\sum _{\pi}{\rm Tr}[\theta \varpi [i|{\cal A}(\pi ;i)]
=\frac{1}{d^2}\sum _{A\subseteq \Omega \setminus \{i\}}\begin{pmatrix}d^2-1\\|A|\\\end{pmatrix}^{-1}{\rm Tr}[\theta \varpi (i|A)]=
\sum _{A\ni i}\frac{{\rm Tr}[\theta {\mathfrak D} (A)]}{|A|}.
\end{eqnarray}
\end{definition}
These formulas are the analogues of Eqs.(\ref{120}), (\ref{12X}), (\ref{1290}).

\subsection{Renormalization of the projectors $\Pi(i)$ into the operators $\sigma (i)$}\label{VI}

Motivated by Eq.(\ref{6nm}) we introduce below the operators $\sigma (i)$, and express the 
Shapley values in terms of them. Later we show that the  $\sigma (i)$ are coherent density matrices.

\begin{proposition}\label{V2}
The following three Hermitian operators are equal to each other, and we denote them as $\sigma (i)$:
\begin{itemize}
\item[(1)]
\begin{eqnarray}\label{eq1}
\sigma (i)=\frac{d}{(d^2!)}\sum _{\pi}\varpi [i|{\cal A}(\pi ;i)]=\frac{d}{(d^2!)}\left [(d^2-1)!\Pi(i)+(d^2-2)!\sum _j\varpi(i|j)+...\right ]
;\;\;\;i\in \Omega.
\end{eqnarray}
Coherent states are here in an order described by a permutation $\pi$. $\varpi [i|{\cal A}(\pi ;i)]$ is
the contribution of the coherent state $\ket{C;i}$ to the projector $\Pi[{\cal A}(\pi ;i)]$ of the preceding coherent states  in the permutation $\pi$.
The summation on the left hand side is over all $d^2!$ permutations.
On the right hand side the first term corresponds to the $(d^2-1)!$ permutations that have $i$ as the first element.
The second term corresponds to the permutations that have $i$ as second element; etc.
\item[(2)]
\begin{eqnarray}\label{eq2}
\sigma (i)=\frac{1}{d}\sum _{A\subseteq \Omega \setminus \{i\}}\begin{pmatrix}d^2-1\\|A|\\\end{pmatrix}^{-1}\varpi (i|A)=
\frac{1}{d}\left [\Pi(i)+\frac{1}{d^2-1}\sum _j\varpi (i|j)+...\right ]
;\;\;\;i\in \Omega.
\end{eqnarray}
Here we add the contributions of the coherent state $i$ to all aggregations (subsets of $\Omega \setminus \{i\}$).
All aggregations with the same cardinality, have equal weight.
Aggregations with different cardinalities are equally likely, and 
 this leads to the coefficient $\begin{pmatrix}d^2-1\\|A|\\\end{pmatrix}^{-1}$. 
\item[(3)]
\begin{eqnarray}\label{eq3}
\sigma (i)=d\sum _{A\ni i}\frac{{\mathfrak D} (A)}{|A|}=d\left [\Pi(i)+\sum _j\frac{{\mathfrak D} (i,j)}{2}+...\right ];\;\;\;i\in \Omega.
\end{eqnarray}
The summation involves subsets $A$ of $\Omega$, which contain $i$.
The first term is ${\mathfrak D}(i)=\Pi(i)$, and the rest are the 'added values' 
quantified with the M\"obius terms ${\mathfrak D}(A)$,
in all the aggregations that the coherent state $i$ participates.
They are divided by $|A|$ because the added value in each aggregation, is distributed equally among all its members.
The M\"obius operators quantify the non-orthogonality and non-commutativity 
of the coherent projectors $\Pi (i)$ (see Eqs.(\ref{m12}),(\ref{375})). Therefore Eq.(\ref{eq3}) shows that
the renormalization of $\Pi (i)$ is related to their non-orthogonality and non-commutativity.
\end{itemize}
The Shapley values in Eq.(\ref{6nm}) can be written as
\begin{eqnarray}\label{zse}
S(i)&=&\frac{1}{d}{\rm Tr}[\theta \sigma (i)];\;\;\;i\in \Omega.
\end{eqnarray}
\end{proposition}
\begin{proof}
Eq.(\ref{6nm}) shows that the traces of these three Hermitian operators with an arbitrary Hermitian operator $\theta$, 
are equal to each other. Therefore these three operators are equal to each other.
\end{proof}
The three formulas in proposition \ref{V2}, describe three different (but equivalent) maps of the coherent projectors
$\Pi(i)$ into the `dressed coherent states\rq{} $\sigma(i)$ (which we prove below to be coherent density matrices). 
The $\sigma (i)$ is equal to $\Pi(i)$ plus the contribution of the coherent state $\ket{C;i}$ to various aggregations of coherent states.
This dressing is related to the fact that the coherent projectors $\Pi(i)$ are non-orthogonal and do not commute.
Indeed the fact that the $\Pi(i)$ are non-orthogonal, implies that the $\varpi(i|j)\ne \Pi(i)$ (and similarly for the other terms),  in Eqs. (\ref{eq1}), (\ref{eq2}).
Also the fact that the $\Pi(i)$ do not commute, implies that the ${\mathfrak D} (i,j)\ne 0$ (and similarly for the other terms), in Eq. (\ref{eq3}).

Consequently the Shapley values are dressed $Q$-functions. 

\subsection{The coherent density matrices $\sigma (i)$ and their properties}

In this section we show that the $\sigma (i)$ are coherent density matrices.
We first introduce the coherent density matrices $\tau(i|k)$,
as an intermediate step towards the  $\sigma (i)$.
They quantify the contribution of $\ket{C;i}$ into all aggregations of coherent states with cardinality $k$.
They are the analogues of the terms in the sum of Eq.(\ref{12X}), in cooperative game theory.
\begin{definition}
$\tau (i|k)$ is the (normalized) sum of the generalized coherent projectors $\varpi(i|A)$, for all $\begin{pmatrix}d^2-1\\k\\\end{pmatrix}$
subsets of $\Omega \setminus\{i\}$, with a given cardinality $k$:
\begin{eqnarray}\label{131}
\tau(i|k)=\begin{pmatrix}d^2-1\\k\\\end{pmatrix}^{-1}\sum _{|A|=k}\varpi(i|A);\;\;\;A\subseteq \Omega \setminus\{i\}.
\end{eqnarray}
This is an `average' of the contribution of $\ket{C;i}$, to aggregations
of coherent states with labels in all  sets $A$, with fixed cardinality $k$.
If $k\ge d$, then $\tau(i|k)=0$. Also 
\begin{eqnarray}
\tau(i|0)=\varpi (i|\emptyset)=\Pi(i).
\end{eqnarray}
\end{definition}
\begin{lemma}
\begin{itemize}
\item[(1)]
The set of $\tau (i|k)$ with $i=1,...,d^2$ (and fixed $k\le d-1$) is a set of coherent density matrices.
As such:
\begin{eqnarray}\label{37}
D(\gamma, \delta)\tau [(\alpha , \beta )|k][D(\gamma, \delta)]^{\dagger}=
\tau [(\alpha +\gamma, \beta +\delta)|k],
\end{eqnarray}
and
\begin{eqnarray}\label{38}
\frac{1}{d}\sum _{i=1}^{d^2}\tau (i|k)={\bf 1};\;\;\;k\le d-1.
\end{eqnarray}
\item[(2)]
The $\tau (i|k)$ commutes with $\Pi(i)$:
\begin{eqnarray}\label{nu}
[\tau (i|k),\Pi(i)]=0.
\end{eqnarray}
\item[(3)]
The $\Pi(i)$ is an eigenprojector of $\tau (i|k)$ with corresponding eigenvalue ${\mathfrak e}(k)$, where
\begin{eqnarray}\label{nu1}
\tau (i|k)\Pi(i)={\mathfrak e}(k)\Pi(i);\;\;\;{\mathfrak e}(k)=\frac{d^2-dk}{d^2-k}.
\end{eqnarray}
\end{itemize}
\end{lemma}

\begin{proof}
\begin{itemize}
\item[(1)]
We first prove Eq.(\ref{37}).
Eq.(\ref{27}) is a bijective map from projectors corresponding to subsets $A$ of $\Omega \setminus \{(\alpha ,\beta)\}$, 
to projectors corresponding to subsets $A+(\gamma, \delta)$ of $\Omega \setminus \{(\alpha +\gamma,\beta +\delta)\}$.
The sets $A$ and $A+(\gamma, \delta)$ have the same cardinality.
The definition of $\tau[(\alpha , \beta )|k]$ involves all projectors corresponding to
subsets of $\Omega \setminus \{(\alpha , \beta )\}$ with cardinality $k$, and with equal weight.
Consequently, the $D(\gamma, \delta)\tau [(\alpha , \beta )|k][D(\gamma, \delta)]^{\dagger}$ involves 
projectors corresponding to all subsets of $\Omega \setminus \{(\alpha +\gamma, \beta +\delta)\}$ with cardinality $k$ and with equal weight, and it
is equal to $\tau[(\alpha +\gamma, \beta +\delta)|k]$.

We next show that the trace of $\tau (i|k)$ is equal to $1$.
The $\tau (i|k)$ are sums of projectors with positive coefficients, and therefore they are Hermitian positive semidefinite operators.
Their trace is
\begin{eqnarray}\label{1210}
{\rm Tr}[\tau(i|k)]=\begin{pmatrix}d^2-1\\k\\\end{pmatrix}^{-1}\sum _{|A|=k}{\rm Tr}[\varpi(i|A)]
;\;\;\;A\subseteq \Omega \setminus \{i\}
\end{eqnarray}
where
\begin{eqnarray}
&&|A|<d\;\;\rightarrow\;\;{\rm Tr}[\varpi(i|A)]=1\nonumber\\
&&|A|\ge d\;\;\rightarrow\;\;{\rm Tr}[\varpi(i|A)]=0.
\end{eqnarray}
For fixed $|A|=k$, the number of possible choices for the set $A$ is $\begin{pmatrix}d^2-1\\k\\\end{pmatrix}$ and
therefore
\begin{eqnarray}
{\rm Tr}[\tau(i|k)]=\begin{pmatrix}d^2-1\\k\\\end{pmatrix}^{-1}\begin{pmatrix}d^2-1\\k\\\end{pmatrix}=1.
\end{eqnarray}
It is now clear that $\tau(i|k)$ are special cases of coherent density matrices.
Therefore the resolutions of the identity in Eq.(\ref{38}), hold.
\item[(2)]
\begin{eqnarray}
\begin{pmatrix}d^2-1\\k\\\end{pmatrix}\tau(i|k)=\sum _{|A|=k}\varpi(i|A)=\sum _{|A|=k}\Pi(A\cup \{i\})-\sum _{|A|=k}\Pi(A);\;\;\;A\subseteq \Omega \setminus\{i\}.
\end{eqnarray}
The resolution of the identity in Eq.(\ref{c45}), that involves all subsets with a given cardinality $k$, can be written as
\begin{eqnarray}
&&\sum _{|B|=k-1}\Pi(B\cup \{i\})+\sum _{|A|=k}\Pi(A)=d\begin{pmatrix}d^2-1\\k-1\\\end{pmatrix}{\bf 1};\;\;\;
A,B\subseteq \Omega \setminus \{i\}
\end{eqnarray}
It follows that 
\begin{eqnarray}\label{nu10}
\begin{pmatrix}d^2-1\\k\\\end{pmatrix}\tau(i|k)=\sum _{|A|=k}\Pi(A\cup \{i\})-d\begin{pmatrix}d^2-1\\k-1\\\end{pmatrix}{\bf 1}
+\sum _{|B|=k-1}\Pi(B\cup\{i\});\;\;\;A,B\subseteq \Omega \setminus\{i\}.
\end{eqnarray}
This and the fact that $\Pi(A\cup \{i\})\Pi(i)=\Pi(i)\Pi(A\cup \{i\})=\Pi(i)$, proves Eq.(\ref{nu}).
\item[(3)]
We multiply both sides of Eq.(\ref{nu10}), with $\Pi(i)$ and we get
\begin{eqnarray}\label{nu104}
\begin{pmatrix}d^2-1\\k\\\end{pmatrix}\tau(i|k)\Pi(i)=\begin{pmatrix}d^2-1\\k\\\end{pmatrix}\Pi(i)-d\begin{pmatrix}d^2-1\\k-1\\\end{pmatrix}\Pi(i)+
\begin{pmatrix}d^2-1\\k-1\\\end{pmatrix}\Pi(i).
\end{eqnarray}
Therefore
\begin{eqnarray}\label{nu105}
\tau(i|k)\Pi(i)=\begin{pmatrix}d^2-1\\k\\\end{pmatrix}^{-1}\begin{pmatrix}d^2-1\\k\\\end{pmatrix}\Pi(i)
-(d-1)\begin{pmatrix}d^2-1\\k\\\end{pmatrix}^{-1}\begin{pmatrix}d^2-1\\k-1\\\end{pmatrix}\Pi(i)={\mathfrak e}(k)\Pi(i).
\end{eqnarray}
where ${\mathfrak e}(k)$ is given in Eq.(\ref{nu1}).
\end{itemize}
\end{proof}

The matrices $\sigma(i)$ can be written in terms of the $\tau(i|k)$ as
\begin{eqnarray}\label{121}
\sigma(i)=\frac{1}{d}[\Pi(i)+\tau(i|1)+...+\tau(i|d-1)].
\end{eqnarray}
\begin{proposition}\label{cxz}
\begin{itemize}
\item[(1)]
The $\sigma(i)$ are coherent density matrices, and therefore they obey the relations
\begin{eqnarray}\label{37av}
D(\gamma, \delta)\sigma(\alpha _i, \beta _i)[D(\gamma, \delta)]^{\dagger}=\sigma(\alpha _i+\gamma , \beta _i+\delta)
\end{eqnarray}
and
\begin{eqnarray}\label{12}
\frac{1}{d}\sum _{i=1}^{d^2} \sigma(i)={\bf 1};\;\;\;\;
{\rm Tr}[\sigma(i)]=1.
\end{eqnarray}
The Shapley values in Eq.(\ref{zse}) are $Q$-functions with respect 
to the coherent density matrices $\sigma (i)$.
\item[(2)]
The $\sigma(i)$ commutes with $\Pi(i)$:
\begin{eqnarray}
[\sigma(i),\Pi(i)]=0.
\end{eqnarray}
\item[(3)]
The $\Pi(i)$ is an eigenprojector of $\sigma(i)$ with corresponding eigenvalue ${\mathfrak e}$, where
\begin{eqnarray}\label{nu4}
\sigma (i)\Pi(i)={\mathfrak e}\Pi(i);\;\;\;{\mathfrak e}=\sum _{k=0}^{d-1}\frac{d-k}{d^2-k}.
\end{eqnarray}
\end{itemize}
\end{proposition}
\begin{proof}
\begin{itemize}
\item[(1)]
The $\sigma(i)$ are sums of the coherent density matrices $\tau(i|k)$, with $1/d$ as coefficients (Eq.(\ref{121}).
It follows that the $\sigma(i)$ are themselves coherent density matrices.
\item[(2)]
This follows immediately from Eq.(\ref{nu}).
\item[(3)]
This follows immediately from Eq.(\ref{nu1}).
\end{itemize}
\end{proof}

Since the Shapley values are $Q$-functions with respect 
to the coherent density matrices $\sigma (i)$,
they have the properties in Eq.(\ref{mmm11}), which are the analogues of the properties in Eqs.(\ref{bbb}),(\ref{add}), in the context of cooperative game theory.  
An analogue to Eq.(\ref{bbbc}) is valid in a quantum context within the set $\Sigma _{VQ}$, but not within the physical set $\Sigma _{Q}$.
As we have seen in proposition \ref{rfv}, the $Q_{\pi}$ corresponding to a physical $Q$-function $Q(i|\theta)$ is a virtual $Q$-function, and there is no operator $\theta _{\pi}$ corresponding to $\theta$.

From proposition \ref{cxz} it follows that
\begin{eqnarray}
&&\sigma(i)={\mathfrak e}\Pi(i)+\sum _{k=1}^{d-1}{\mathfrak e} _kP_k(i);\;\;\;\sum _{k=1}^{d-1}{\mathfrak e} _k=1-{\mathfrak e}\nonumber\\
&&P_k(i)P_{\ell}(i)=P_k(i)\Pi(i)=0.
\end{eqnarray}
where ${\mathfrak e} _k, P_k (i) $ are eigenvalues and eigenprojectors of $\sigma(i)$, which for a given fiducial vector, are defined uniquely. 
The 
\begin{eqnarray}
\sigma(i)-\Pi(i)=\sum _{k=1}^{d-1}{\mathfrak e} _k[P_k(i)-\Pi(i)]
\end{eqnarray}
It is easily seen that
\begin{eqnarray}
\sum _{i=1}^{d^2}[\sigma(i)-\Pi(i)]=0;\;\;\;{\rm Tr}[\sigma(i)-\Pi(i)]=0.
\end{eqnarray}
Therefore
\begin{eqnarray}
\sum _{i=1}^{d^2}S(i|\theta)=\sum _{i=1}^{d^2}Q(i|\theta)={\rm Tr}\theta.
\end{eqnarray}
Also for $\theta ={\bf 1}$, we get 
\begin{eqnarray}
S(i|{\bf 1})=Q(i|{\bf 1})=\frac{1}{d}.
\end{eqnarray}
\subsection{Example}\label{exa}
In the three-dimensional Hilbert space $H(3)$, we consider the $9$ coherent states $\ket{C; \alpha, \beta}$ 
with the generic fiducial vector
\begin{eqnarray}\label{fi}
&&\ket{\eta }=a\ket{X;0}+b\ket{X;1}+c\ket{X;2}\nonumber\\
&&a=0.169(1+i);\;\;\;\;b=-0.338;\;\;\;\;c=0.845-0.338i
\end{eqnarray}
There are $9$ projectors $\Pi (i)$ and $36$ projectors $\Pi(i,j)$  with $i,j=1,...,9$ and $i<j$.
The projectors $\Pi(i,j,k)={\bf 1}$.
The coherent density matrix $\sigma (1)$ is given by
\begin{eqnarray}
&&\sigma(1)=\frac{1}{3}\Pi(1)+\frac{1}{24}\sum _{i=2}^9\varpi(1|i)+\frac{1}{84}\sum _{i,j=2}^9\varpi(1|i,j);\;\;\;i<j\nonumber\\
&&\varpi(1|i)=\Pi(1,i)-\Pi(i);\;\;\;\;\varpi(1|i,j)=\Pi(1,i,j)-\Pi(i,j)={\bf 1}-\Pi(i,j)
\end{eqnarray}
There are $8$ and $28$ terms, in the two sums on the right hand side of this equation, and we get
\begin{eqnarray}
\sigma(1)=\begin{pmatrix}
0.162&-0.040-0.038i&0.049+0.117i\\
-0.040+0.038i&0.210&-0.164-0.065i\\
0.049-0.117i&-0.164+0.065i&0.628
\end{pmatrix},
\end{eqnarray}
in the position basis.
Its eigenvalues are
\begin{eqnarray}\label{eig}
{\mathfrak e}=0.726;\;\;\;\;{\mathfrak e} _1=0.125;\;\;\;\;{\mathfrak e} _2=0.149.
\end{eqnarray}
The value of ${\mathfrak e}$ agrees with the result given in Eq.(\ref{nu4}), and it does not depend on the choice of the fiducial vector.
The values of ${\mathfrak e} _1, {\mathfrak e} _2$ do depend on the choice of the fiducial vector.
It is seen that  the density matrix $\sigma(1)$ describes a mixed state.
The rest of the density matrices $\sigma(i)$ have been calculated through displacement transformations (Eq.(\ref{37av})), which we rewrite
 in the single index notation, as
\begin{eqnarray}
D(i)\sigma(1)[D(i)]^{\dagger}=\sigma(i).
\end{eqnarray}
Since they are related to $\sigma (1)$ through these unitary transformations, they have the same eigenvalues, given in Eq.(\ref{eig}).

We have calculated the Shapley values in Eq.(\ref{zse}), for the following  operators $\theta$:
\begin{eqnarray}\label{example}
\theta_1=\ket{X;0}\bra{X;0};\;\;\;\theta_2=\Pi(6)=\Pi(1,2);\;\;\;\theta _3=\Pi(1,6)=\Pi(0,0;1,2)
;\;\;\;\theta_4=
\begin{pmatrix}
3&1-i&-2\\
1+i&5&2-i\\
-2&2+i&4
\end{pmatrix}
\end{eqnarray}
The matrix $\theta_4$ is expressed in the position basis.
The results are shown in table \ref{t1}.
We also give in table \ref{t2} the values of the $Q$-function in Eq.(\ref{mmm}), for these operators.
The $\delta (i)=S(i)-Q(i)$ 
describes the dressing of $Q(i)$, due to the contribution of the coherent state $\ket{C;i}$ into aggregations of other coherent states.

\subsection{Equivalence classes in $\Sigma _{VQ}$}\label{equiv}
In $\Sigma _G$, we use the notation $v_1\underset {S}\sim v_2$ for two cooperative games with the same Shapley values $S(i|v_1)=S(i|v_2)$.
It is easily seen that $\underset {S}\sim$ is an equivalence relation.
We define several equivalence relations in the paper, and we use indices to distinguish them.
Here the index $S$ in the notation indicates `Shapley'.

The properties reflexivity, symmetry and transitivity, hold:
\begin{eqnarray}
&&v_1\underset {S}\sim v_1\nonumber\\
&&v_1\underset {S}\sim v_2\;\rightarrow\; v_2\underset {S}\sim v_1\nonumber\\
&&v_1\underset {S}\sim v_2\;{\rm and}\;v_2\underset {S}\sim v_3\;\rightarrow\; v_1\underset {S}\sim v_3
\end{eqnarray}
The set $\Sigma _G$ is partitioned into equivalence classes.
Cooperative games in the same equivalence class have the same Shapley values.
$\Sigma _G/\underset {S}\sim$ is the set of these equivalence classes.

In analogous way the set $\Sigma _{VQ}$ is partitioned into equivalence classes of virtual $Q$-functions, with the same Shapley values.
The following proposition shows that in each equivalence class, there is exactly one physical $Q$-function and exactly one Hermitian operator $\theta$, with these Shapley values.
\begin{proposition}
A Hermitian operator is defined uniquely by its Shapley values.
\end{proposition}
\begin{proof}
This is seen from the fact that 
for a given set of values $\{{\cal S}(i)\}$, Eqs(\ref{zse}) form a system of $d^2$ equations with $d^2$ unknowns (the elements of the Hermitian matrix $\theta$).
In the single index notation,
we rewrite the $d\times d$  matrices $\sigma (i)$ as ${\mathfrak s}_{ij}$, and the $d\times d$  matrix $\theta$ as a vector $\theta _{j}$
(where $j=1,...,d^2$), and we get 
\begin{eqnarray}
S(i)=\frac{1}{d}\sum _j{\mathfrak s}_{ij}\theta _j;\;\;\;i,j\in \Omega.
\end{eqnarray}
The generic nature of the fiducial vector ensures that the determinant of ${\mathfrak s}_{ij}$ is non-zero.
We solve this system and we find the operator $\theta$ (and then we can calculate all the $Q(A|\theta)$).
\end{proof}

We use the notation $[\theta]$ for the equivalence class of $Q$-functions that contains the $Q$-function of the operator $\theta$.
It is easy to verify that 
\begin{eqnarray}
\lambda _1[\theta _1]+\lambda _2[\theta _2]=[\lambda _1\theta _1+\lambda _2\theta _2].
\end{eqnarray}
Therefore $\Sigma _{VQ}/\underset {S}\sim$ is isomorphic to $\Sigma _{Q}$:
\begin{eqnarray}
 \Sigma _{Q}\cong \Sigma _{VQ}/\underset {S}\sim .
\end{eqnarray}

\section{Generalization to total sets of states}\label{gen}

The work has been presented in the context of coherent states. However it could be generalized to any total set of states.
In the space $H(d)$, we consider $n>d$ states (which are not in general coherent states): 
\begin{eqnarray}
\Sigma=\{\ket{v_i}|i=1,...,n;\;n> d\}
\end{eqnarray}
Any subset of $d$ of these states are assumed to be linearly independent. Any subset of $k\ge d$ of these states, is a total set.
We do not have a resolution of the identity in terms of these $n$ states.
There is merit in using $\{\ket{v_i}\}$ instead of orthonormal bases, in the study of problems.
There is redundancy in them, which is important for error correction, and which is absent in orthonormal bases.
But they can only be practically useful if we have a resolution of the identity in terms of them.

The subject of matroids\cite{M1,M2,M3} provides a deep aproach to linear dependence. In this section, we  complement this with a method that replaces $\Sigma$ with
the following set of renormalized density matrices, which resolve the identity:
\begin{eqnarray}
\Sigma _{\rm ren}=\{\sigma (i)|i=1,...,n;\;n> d\};\;\;\;\frac{d}{n}\sum _{i=1}^n\sigma (i)={\bf 1}.
\end{eqnarray}
We renormalize the projectors $\Pi(i)=\ket{v_i}\bra{v_i}$ into density matrices $\sigma (i)$, using Eqs(\ref{eq2}),(\ref{eq3}), which are written in the present notation as
\begin{eqnarray}\label{eq20}
\sigma  (i)=\frac{d}{n}\sum _{A\subseteq \Omega \setminus \{i\}}\begin{pmatrix}n-1\\|A|\\\end{pmatrix}^{-1}\varpi (i|A)=
\frac{d}{n}\left [\Pi(i)+\frac{1}{n-1}\sum _j\varpi (i|j)+...\right ]
;\;\;\;i\in \Omega.
\end{eqnarray}
and
\begin{eqnarray}\label{eq30}
\sigma  (i)=\frac{n}{d}\sum _{A\ni i}\frac{{\mathfrak D} (A)}{|A|}=\frac{n}{d}\left [\Pi (i)+\sum _j\frac{{\mathfrak D} (i,j)}{2}+...\right ];\;\;\;i\in \Omega.
\end{eqnarray}
Here $\Omega$ is the set $\{1,...,n\}$ (the $d^2$ is replaced by $n$).
The $\varpi(i|A)$,  ${\mathfrak D} (A)$ are defined in analogous way to that described in sections \ref{IVB}, \ref{IVA}, correspondingly.
The $\Pi(i)$ and $\sigma (i)$ are general (not coherent) projectors and density matrices correspondingly.
Here the $\sigma(i)$, $\Pi(i)$ do not commute in general.

If the set $\Sigma$ is an orthonormal basis of $d$ vectors then $\sigma  (i)=\Pi (i)$.
\begin{example}
In $H(2)$ we consider the total set of states:
\begin{eqnarray}
\Sigma=\left \{\ket{X;0}, \frac{1}{\sqrt 5}(\ket{X;0}+2\ket{X;1})\right\}.
\end{eqnarray}
In this case $n=2$, and
\begin{eqnarray}
&&\Pi(1)=\begin{pmatrix}
1&0\\
0&0\\
\end{pmatrix}\rightarrow \sigma  (1)=\frac{1}{10}
\begin{pmatrix}
9&-2\\
-2&1\\
\end{pmatrix};\nonumber\\
&&\Pi(2)=\frac{1}{5}\begin{pmatrix}
1&2\\
2&4\\
\end{pmatrix}\rightarrow \sigma  (2)=\frac{1}{10}
\begin{pmatrix}
1&2\\
2&9\\
\end{pmatrix}
\end{eqnarray}
The resolution of the identity is $\sigma  (1)+\sigma  (2)={\bf 1}$.
\end{example}
\begin{example}
In $H(2)$ we consider the total set of states:
\begin{eqnarray}
\Sigma=\left \{\ket{X;0}, \ket{X;1}, \frac{1}{\sqrt 5}(\ket{X;0}+2\ket{X;1})\right\}.
\end{eqnarray}
In this case $n=3$, and
\begin{eqnarray}
{\mathfrak D}(1,2)=0;\;\;\;
{\mathfrak D}(1,3)=\frac{1}{5}\begin{pmatrix}
-1&-2\\
-2&1\\
\end{pmatrix};\;\;\;
{\mathfrak D}(2,3)=\frac{1}{5}\begin{pmatrix}
4&-2\\
-2&4\\
\end{pmatrix};\;\;\;
{\mathfrak D}(1,2,3)=\frac{1}{5}\begin{pmatrix}
-4&2\\
2&-1\\
\end{pmatrix}.
\end{eqnarray}
Therefore
\begin{eqnarray}
&&\Pi(1)=\begin{pmatrix}
1&0\\
0&0\\
\end{pmatrix}\rightarrow \sigma  (1)=\frac{1}{20}
\begin{pmatrix}
19&-2\\
-2&1\\
\end{pmatrix};\nonumber\\
&&\Pi(2)=\begin{pmatrix}
0&0\\
0&1\\
\end{pmatrix}\rightarrow \sigma  (2)=\frac{1}{20}
\begin{pmatrix}
4&-2\\
-2&16\\
\end{pmatrix}\nonumber;\\
&&\Pi(3)=\frac{1}{5}\begin{pmatrix}
1&2\\
2&4\\
\end{pmatrix}\rightarrow \sigma  (3)=\frac{1}{20}
\begin{pmatrix}
7&4\\
4&13\\
\end{pmatrix}.
\end{eqnarray}
We note the analogy between this example, and the example in section \ref{100}.
The resolution of the identity is
\begin{eqnarray}
\frac{2}{3}[\sigma  (1)+\sigma  (2)+\sigma  (3)]={\bf 1}.
\end{eqnarray}
\end{example}

\subsection{The Shapley methodology in cooperative game theory and its application in a quantum context}

We compare and contrast briefly, the Shapley methodology in cooperative game theory, with our application of this methodology in a quantum context.
\begin{itemize}
\item
The Shapley methodology in cooperative game theory renormalizes the values $v(i)$ of the various players, into the Shapley values ${\cal S}(i)$.
This takes into account the contribution of a player into coalitions, and is related to the fact that $v(i,j,...)\ne v(i)+v(j)+...$ (Eq.(\ref{149})).
Our formalism renormalizes  the projectors $\Pi(i)$ into $\sigma(i)$.
This takes into account the contribution of a state into aggregations of many states, and is related to 
the fact that $\Pi(i,j,...)\ne \Pi(i)+\Pi(j)+...$ (Eq.(\ref{123})).

\item
In cooperative game theory the `total worth\rq{} is shared among the players and $\sum {\cal S}(i)=v(N)$.
The $\sum v(i)$ might or might not be equal to $v(N)$.
In our formalism the identity ${\bf 1}$ is `shared\rq{} among various states in $\Omega$ and we have the resolution of the identity $\frac{d}{n}\sum \sigma (i)={\bf 1}$. 
The $\sum \Pi (i)$ might (the case with coherent projectors) or might not (the case considered in this section) be equal to ${\bf 1}$.

\item
If for all coalitions $v(i,j,...)= v(i)+v(j)+...$ then ${\cal S}(i)=v(i)$.
In a quantum context, the orthonormal bases obey the relation $\Pi(i,j,...)= \Pi(i)+\Pi(j)+...$, and then  $\sigma (i)=\Pi(i)$.
\end{itemize}

\section{Applications}\label{appl}

\subsection{Comonotonic Hermitian  positive semidefinite operators}

We order the values of the $Q$-function of a Hermitian  positive semidefinite operator $\theta$ as follows:
\begin{eqnarray}
Q(i_1|\theta)\ge Q(i_2|\theta)\ge ...\ge Q(i_{d^2}|\theta)\ge 0.
\end{eqnarray}
We call the $(i_1,i_2,...,i_{d^2})_{\theta, Q}$ `location index\rq{} of $\theta$. 
The operator $\theta$ `lives\rq{} primarily at the point $i_1$ in the ${\mathbb Z}(d)\times {\mathbb Z}(d)$
phase space, to a lesser extent at the point $i_2$, to even lesser extent at the point $i_3$, etc.
In a similar way we define the $(i_1,i_2,...,i_{d^2})_{\theta, S}$, with respect to the Shapley values.
In tables \ref{t1}, \ref{t2}, we give the locations indices for the operators in Eq.(\ref{example}).
The two location indices are in general different (e.g., compare the location indices for the operator $\theta_4$). 

In the case that some values of the $Q$-function are equal to each other, we use braces in the notation.
For example, the location index for $\theta _1$ in table \ref{t1} is $(\{2,5,8\},\{3,6,9\},\{1,4,7\})_{\theta _1,Q}$,
because $Q(2)=Q(5)=Q(8)$, and $Q(3)=Q(6)=Q(9)$, and $Q(1)=Q(4)=Q(7)$.
Similar notation is used for the case where some of the Shapley values are equal to each other.
 
We next consider the set  ${\cal M}_n(Q)$ (${\cal M}_n(S)$) of operators for which 
the $n$ largest values of the $Q$-function  (Shapley function) are different from each other.
They are both subsets of the set of all Hermitian  positive semidefinite operators. 
\begin{definition}
Two operators $\theta, \phi \in {\cal M}_n(Q)$ are called $(n,Q)$-comonotonic or $(n,Q)$-cohabitant, if 
the first $n$ integers in the location index of $\theta$, are equal to the first $n$ integers in the location index of $\phi$,
with respect to the $Q$-function
\begin{eqnarray}
i_k({\theta, Q})=i_k({\phi, Q});\;\;\;k=1,...,n.
\end{eqnarray}
We denote this as $\theta \underset {n, Q}\sim \phi$.  
In a similar way we define $(n,S)$-comonotonic or $(n,S)$-cohabitant operators in ${\cal M}_n(S)$, with respect to the Shapley values.
\end{definition}
This definition is motivated by the expectation that $n$-cohabitant operators will have similar physical properties.
Within ${\cal M}_n(Q)$ transitivity holds, and $\underset {n,Q}\sim$ is an equivalence relation.
Similarly, $\underset {n,S}\sim$ is an equivalence relation in ${\cal M}_n(S)$.

In ref.\cite{v16b}, we have defined (using another equivalent definition) what in the present notation is $(d,Q)$-comonotonicity, and used it in connection with Choquet integrals.
Here we define the $(n,Q)$-comonotonicity for any $n$, and we also define the $(n,S)$-comonotonicity with respect to the Shapley values.
$n$ describes the `strength of cohabitation\rq{} (if $m\ge n$, the $m$-comonotonicity is stronger concept than the $n$-comonotonicity).
\begin{remark}
Above we considered operators in ${\cal M}_n(Q)$.
If we define $\underset {n,Q}\sim$ in the larger set ${\mathfrak H}(d)$, then transitivity might not hold.
For example, let $\theta, \phi, \psi$ be three operators for which the two largest values of the
$Q$-function, are $Q(1)$ and $Q(2)$. We also assume that
\begin{eqnarray}
Q(1|\theta)>Q(2|\theta);\;\;\;Q(1|\phi)=Q(2|\phi);\;\;\;Q(1|\psi)<Q(2|\psi)
\end{eqnarray}
In this example, $\theta \underset {2,Q}\sim \phi$ and $\phi \underset {2,Q}\sim \psi$, but it is not true that
$\theta \underset {2,Q}\sim \psi$.
Similar comment can be made for ${\cal M}_n(S)$.
\end{remark}

\subsection{Relocation of a Hamiltonian in phase space and physical consequences}\label{nm}

In many applications we have an operator $\theta (\lambda)$ which is a continuous function of a coupling constant $\lambda$.
One example is 
a Hamiltonian $\theta (\lambda)=\theta _1+\lambda \theta _2$, where $\theta _1$ is the free part, $\theta _2$ the interaction part, and $\lambda$ the coupling constant.
Another example is the
$\ket{g(\lambda)}\bra{g(\lambda)}$ where $\ket{g(\lambda)}$ is the eigenstate of a Hamiltonian
$\theta (\lambda)$ corresponding to the lowest eigenvalue (ground state).

We have defined in ref.\cite{v16b} comonotonicity intervals of the coupling constant in such operators, and we used them with Choquet integrals.
Below we show that they are important in their own right, and that they can be used to quantify the concept of relocation of a Hamiltonian in phase space, as the coupling constant varies.

\begin{definition}
Let $\theta (\lambda)$ be an operator that depends on a real parameter $\lambda$.
${\cal R}\subseteq {\mathbb R}$ is `$(n,Q)$-comonotonicity interval\rq{}, or `$(n,Q)$-cohabitation interval\rq{}, if for any pair $\lambda _1, \lambda _2$ in ${\cal R}$,
the $\theta (\lambda _1)$, $\theta (\lambda _2)$ are $(n,Q)$-comonotonic.
In a similar way, using the Shapley values, we define a `$(n,S)$-comonotonicity interval\rq{}, or `$(n,S)$-cohabitation interval\rq{}.
\end{definition}
We will use these ideas with Hamiltonians $\theta (\lambda)$ that depend on a coupling constant $\lambda$.
We show with examples, that as $\lambda$ changes within a cohabitation interval, physical quantities related to this Hamiltonian 
(e.g., the ground state of the system), change slowly. 
When $\lambda$ crosses from one cohabitation interval to another, the Hamiltonian 
relocates from one region to another in phase space, and physical quantities change more drastically.
\begin{example}\label{ex12}
We consider the following Hamiltonian in the Hilbert space $H(3)$
\begin{eqnarray}
\theta(\lambda)={\bf 1}+\lambda \ket{X;0}\bra{X;0}
\end{eqnarray}
where $\lambda$ is a coupling constant.
The eigenvalues of this Hamiltonian are $\kappa_1=1+\lambda$, $\kappa _2=\kappa _3=1$.
For negative $\lambda$, the lowest eigenvalue is $\kappa_1$, and the corresponding eigenstate (`ground state')
$\ket{X;0}$. For positive $\lambda$, the lowest eigenvalues are $\kappa _2=\kappa _3=1$
and the corresponding eigenstates $\alpha \ket{X;1}+\beta \ket{X;2}$, with arbitrary $\alpha, \beta$.
Clearly, a drastic change in the ground state occurs at $\lambda=0$, which in a large (ideally infinite) system would be a phase transition.

The Shapley values $S(i)$, and the $Q$-values $Q(i)$ of Eq.(\ref{mmm}),  are
\begin{eqnarray}
S(i)=\frac{1}{3}+\lambda S(i|\ket{X;0}\bra{X;0});\;\;\;Q(i)=\frac{1}{3}+\lambda Q(i|\ket{X;0}\bra{X;0})
\end{eqnarray}
where $S(i|\ket{X;0}\bra{X;0})$ and $Q(i|\ket{X;0}\bra{X;0})$ are given in the first row in tables \ref{t1}, \ref{t2}, correspondingly.
As $\lambda$ changes from negative to positive values, the order within
the location index of the Hamiltonian is reversed (with respect to both the Shapley values and the $Q$ values):
\begin{eqnarray}
(\{1,4,7\},\{3,6,9\}, \{2,5,8\})\;\rightarrow\;(\{2,5,8\},\{3,6,9\},\{1,4,7\}).
\end{eqnarray}
It is seen that relocation of the Hamiltonian in phase space, is associated with large changes in the ground state.

In a similar way consider the following Hamiltonian
\begin{eqnarray}
\theta(\lambda)={\bf 1}+\lambda \ket{C;1,2}\bra{C;1,2}
\end{eqnarray}
The eigenvalues of this Hamiltonian are $\kappa_1=1+\lambda$, $\kappa _2=\kappa _3=1$, and the eigenstates
$\ket{C;1,2}$ and $\alpha \ket{s_1}+\beta \ket{s_2}$, correspondingly. Here
$\ket{s_1},\ket{s_2}$ are two vectors perpendicular to the coherent state $\ket{C;1,2}$.
In this example also, for negative $\lambda$ the lowest eigenvalue is $\kappa_1$, and for positive $\lambda$ it is $\kappa _2=\kappa _3=1$.
Therefore the ground state changes drastically as we go from negative to positive values of $\lambda$.

The Shapley values $S(i)$ and the $Q$-values $Q(i)$,  are
\begin{eqnarray}
S(i)=\frac{1}{3}+\lambda S(i|\ket{C;1,2}\bra{C;1,2});\;\;\;Q(i)=\frac{1}{3}+\lambda Q(i|\ket{C;1,2}\bra{C;1,2})
\end{eqnarray}
where $S(i|\ket{C;1,2}\bra{C;1,2})$ and $Q(i|\ket{C;1,2}\bra{C;1,2})$ are given in the second row in tables \ref{t1}, \ref{t2}, correspondingly.
As $\lambda$ changes from negative to positive values, the order within 
the location index of the Hamiltonian is reversed (with respect to both the Shapley values and the $Q$ values):
\begin{eqnarray}
(2,7,8,1,5,4,3,9,6)\;\rightarrow\;(6,9,3,4,5,1,8,7,2).
\end{eqnarray}
In this example also, relocation of the Hamiltonian in phase space is associated with large changes in the ground state.

We note that in these examples, both the Shapley values and the $Q$-values lead to the same conclusions.
This is a desirable feature, because it shows that the method is stable, when we change the definition of the $Q$-function.
However in the next example, the location indices based on Shapley values are better correlated with changes in the ground state of the system,
than the location indices based on $Q$-values.
\end{example}

\begin{example}\label{exa12}
We consider the following Hamiltonian in the Hilbert space $H(3)$
\begin{eqnarray}\label{ham}
\theta(\lambda)=
\begin{pmatrix}
11+2\lambda&(1-\lambda )-i(1-2\lambda)&3\lambda\\
(1-\lambda )+i(1-2\lambda)&8+4\lambda&-\lambda\\
3\lambda&-\lambda&13
\end{pmatrix}.
\end{eqnarray}
The matrix is expressed in the position basis, and $\lambda$ is a coupling constant.
We have calculated the Shapley values $S(i)$ of Eqs(\ref{zse})  and ordered them, for various values of $\lambda$
(for the coherent states we used the fiducial vector in Eq.(\ref{fi})).
The location indices $(i_1,...,i_{d^2})_{\theta, S}$ (using the Shapley values) and 
$(i_1,...,i_{d^2})_{\theta, Q}$ (using the $Q$-values),
are shown in table \ref{t3}. These two types of location indices differ only for $\lambda =0.25$.

We studied the ground state  of this Hamiltonian, i.e., its eigenstate corresponding to the lowest eigenvalue $h_1$:
\begin{eqnarray}
\theta (\lambda)\ket{g(\lambda)}=h_1\ket{g(\lambda)}
\end{eqnarray}
Table \ref{t3}, shows the overlap
$|\langle g(-0.75)\ket{g(\lambda)}|$
between the eigenvectors $\ket{g(\lambda)}$ and $\ket{g(-0.75)}$ (which is taken as a reference vector). It also shows the eigenvalues $h_1,h_2,h_3$.

With respect to the Shapley values, the interval $[-0.75, -0.5]$ and also the interval $[-0.25, 0.25]$ are both $(1,S)$-comonotonicity intervals.
For $\lambda \in (-0.5,-0.25)$, the Hamiltonian relocates in phase space.
Relocations occur also when $\lambda$ varies in the interval $(0.25, 0.75)$.
This is well correlated with the fact that
there is a significant change in the value of $|\langle g(-0.75)\ket{g(\lambda)}|$, in the region from $\lambda=0.25$ to $\lambda =0.75$.

With respect to the $Q$-function, the interval $[-0.75, -0.5]$ and the interval $[-0.25, 0]$ are both $(1,Q)$-comonotonicity intervals.
The interval $[0.25, 0.5]$ is a $(2,Q)$-comonotonicity interval. The Hamiltonian relocates in phase space when $\lambda$ varies in the intervals 
$(-0.5, -0.25)$, and $(0, 0.25)$, and $(0.5, 0.75)$. 
Here the correlation between the relocation of the Hamiltonian in phase space, and changes in its ground state, is not as strong as above.
The interval $[0.25, 0.5]$ is a $(2,Q)$-comonotonicity interval, and yet 
there is significant change in the value of $|\langle g(-0.75)\ket{g(\lambda)}|$ between $\lambda =0.25$ to $\lambda =0.5$.
In this example, the location indices based on Shapley values are better correlated with changes in the ground state of the system,
than the location indices based on $Q$-values.
\end{example}
Overall, the location indices and the comonotonicity intervals, describe the
relocation of the Hamiltonian in phase space as the coupling constant varies, and this is linked to large changes in physical quantities like the ground state. 
In examples \ref{ex12} this link is seen equally well with both the Shapley values and the $Q$-function, but in example \ref{exa12} 
this link is seen better with the Shapley values, than with the $Q$-function.
Further work is needed here, in order to assess the relative merits of the use of relocation indices based on Shapley values or on $Q$-values,
for the study of changes in the ground state.
Application of the method to higher dimensional systems, can provide insight to phase transitions, chaos, etc.

We note that the Wehrl entropy\cite{We} of the $Q$-function, has been used for the study of ground state of quantum systems(e.g., \cite{We1}).
The Wehrl entropy of a positive semidefinite Hermitian operator $\theta$, normalized so that ${\rm Tr}\theta=1$, is given by
$-\sum Q(i|\theta)\ln Q(i|\theta)$, and it is invariant under permutations of the values of the $Q$-function.
Any change in the location of the Hamiltonian in phase space (as defined above), which does not change the set of values $\{Q(i|\theta)\}$,
does not change the Wehrl entropy. Our approach is complementary to the Wehrl entropy, because it is 
based on the location of the Hamiltonian in phase space, which the Wehrl entropy is not able to detect.

\section{Discussion}\label{D}

We have introduced coherent density matrices, which are mixed states that resolve the identity, and which have a closure property 
where under displacements they are transformed into other coherent density matrices. 
We also introduced coherent projectors of rank $n$, to spaces spanned by aggregations of $n$ coherent states ($n\le d$).

We then used the Shapley methodology in cooperative game theory, to renormalize (dress) the coherent projectors $\Pi(i)=\ket{C;i}\bra{C;i}$
into the coherent density matrices $\sigma (i)$ given in proposition \ref{V2}.
The formalism adds to $\Pi(i)$ the contribution of the coherent state $\ket{C;i}$ into aggregations of several coherent states, expressed in terms of
the M\"obius transformations ${\mathfrak D}(A)$ or the projectors $\varpi (i|A)$. 
Consequently the $Q$-function $Q(i)$, is renormalized into the Shapley values $S(i)$.

We also used this methodology with an arbitrary total set of  states, for which we have no resolution of the identity.
The redundancy related to the `total set\rq{} property of the set, is desirable because it can lead to error correction.
But for practical applications a resolution of the identity is needed.
The renormalization formalism, leads to density matrices that resolve the identity.

As an application we studied in section \ref{nm}
the relocation of a Hamiltonian in phase space as the coupling constant varies, and
its effect on the ground state of the system.
We found that in some cases (e.g., example \ref{exa12}) the technique works better with the Shapley values than with the $Q$ function,
but further work is needed for formal results in this direction.

 More generally, the Shapley values in a quantum context is a novel addition to phase space methods, which needs further study.

\newpage
\begin{table}
\caption{The Shapley values of Eq.(\ref{zse}), for the operators $\theta$ given in Eq.(\ref{example}).
The location indices are also shown.}
\def\arraystretch{3}
\begin{tabular}{|c|c|c|c|c|c|c|c|c|c|c|}\hline
&$S(1|\theta)$&$S(2|\theta)$&$S(3|\theta)$&$S(4|\theta)$&$S (5|\theta)$&$S(6|\theta)$&$S(7|\theta)$&$S(8|\theta)$&$S(9|\theta)$&$(i_1,...,i_9)_{\theta,S}$\\\hline
$\theta_1$&0.054&0.209&0.070&0.054&0.209&0.070&0.054&0.209&0.070&$(\{2,5,8\},\{3,6,9\},\{1,4,7\})$\\\hline
$\theta_2$&0.083&0.054&0.155&0.089&0.084&0.241&0.057&0.081&0.156&$(6,9,3,4,5,1,8,7,2)$\\\hline
$\theta_3$&0.287&0.153&0.224&0.260&0.135&0.284&0.224&0.213&0.219&$(1,6,4,\{3,7\},9,8,2,5)$\\\hline
$\theta_4$&1.106&1.481&1.604&1.390&0.996&1.210&1.551&1.058&1.601&$(3,9,7,2,4,6,1,8,5)$\\\hline
\end{tabular} \label{t1}
\end{table}

\begin{table}
\caption{The $Q$-values of Eq.(\ref{mmm}), for the operators $\theta$ given in Eq.(\ref{example}).The location indices are also shown.}
\def\arraystretch{3}
\begin{tabular}{|c|c|c|c|c|c|c|c|c|c|c|}\hline
&$Q(1|\theta)$&$Q(2|\theta)$&$Q(3|\theta)$&$Q(4|\theta)$&$Q (5|\theta)$&$Q(6|\theta)$&$Q(7|\theta)$&$Q(8|\theta)$&$Q(9|\theta)$&$(i_1,...,i_9)_{\theta,Q}$\\\hline
$\theta_1$&0.019&0.276&0.038&0.019&0.276&0.038&0.019&0.276&0.038&$(\{2,5,8\},\{3,6,9\},\{1,4,7\})$\\\hline
$\theta_2$&0.063&0.016&0.184&0.068&0.068&0.333&0.016&0.063&0.184&$(6,\{9,3\},\{4,5\},\{1,8\},\{7,2\})$\\\hline
$\theta_3$&0.333&0.105&0.226&0.282&0.072&0.333&0.222&0.208&0.215&$(\{1,6\},4,3,7,9,8,2,5)$\\\hline
$\theta_4$&0.933&1.609&1.781&1.429&0.776&1.120&1.693&0.842&1.813&$(9,3,7,2,4,6,1,8,5)$\\\hline
\end{tabular} \label{t2}
\end{table}
\begin{table}
\caption{The location indices $(i_1,...,i_{d^2})_{\theta, S}$ (using the Shapley values) and 
$(i_1,...,i_{d^2})_{\theta, Q}$ (using the $Q$-values in Eq.(\ref{mmm})), for the Hamiltonian of Eq.(\ref{ham}) as a function of $\lambda$ . Its eigenvalues
$h_1,h_2,h_3$ and the overlap $|\langle g(-0.75)\ket{g(\lambda)}|$, are also shown.}
\def\arraystretch{2}
\begin{tabular}{|c|c|c|c|c|c|c|c|}\hline
$\lambda$&$(i_1,...,i_{d^2})_{\theta,S}$ &$(i_1,...,i_{d^2})_{\theta, Q}$ &$|\langle g(-0.75)\ket{g(\lambda)}|$&$h_1$&$h_2$&$h_3$\\\hline
$-0.75$&$(4,2,7,1,5,8,9,3,6)$&$(4,2,7,1,5,8,9,3,6)$&$1$&$3.21$&$10.02$&$14.25$\\\hline
$-0.5$&$(4,7,2,1,5,8,9,3,6)$&$(4,7,2,1,5,8,9,3,6)$&$0.998$&$4.67$&$10.60$&$13.71$\\\hline
$-0.25$&$(7,4,1,2,8,5,9,3)$&$(7,4,1,2,8,5,9,3)$&$0.993$&$6.09$&$11.16$&$13.24$\\\hline
$0$&$(7,1,4,2,8,5,9,3)$&$(7,1,4,2,8,5,9,3)$&$0.978$&$7.43$&$11.56$&$13.00$\\\hline
$0.25$&$(7,1,4,8,2,5,9,3)$&$(1,7,4,8,2,5,9,3)$&$0.943$&$8.66$&$11.51$&$13.32$\\\hline
$0.5$&$(1,7,8,5,4,2,9,3)$&$(1,7,8,5,4,2,9,3)$&$0.838$&$9.62$&$11.29$&$14.08$\\\hline
$0.75$&$(8,1,7,5,6,4,9,2,3)$&$(8,1,7,5,6,4,9,2,3)$&$0.535$&$9.90$&$11.51$&$15.08$\\\hline
\end{tabular} \label{t3}
\end{table}

\end{document}